\documentclass[a4paper,USenglish,cleveref,thm-restate]{lipics-v2021}
\usepackage{nicefrac}
\usepackage{xspace}
\usepackage{mathtools}
\usepackage{complexity}
\usepackage{microtype}
\usepackage{placeins}
\newcommand{\old}[1]{{}}

\newcommand{\revised}[1]{{#1}}
\newcommand{\new}[1]{\textcolor{black}{#1}}

\crefname{figure}{Figure}{Figures}
\crefname{theorem}{Theorem}{Theorems}
\crefname{lemma}{Lemma}{Lemmas}
\crefname{claim}{Claim}{Claims}
\crefname{observation}{Observation}{Observations}
\crefname{corollary}{Corollary}{Corollaries}
\crefname{section}{Section}{Sections}

\title{Packing Squares into a Disk with Optimal Worst-Case Density\thanks{An extended abstract appears in the Proceedings
of the 37th Symposium on Computational Geometry (SoCG 2021)~\cite{fgj+-psdow-20}.}}

\titlerunning{Worst-Case Optimal Squares Packing into Disks}

\author{S\'{a}ndor P. Fekete}{Department of Computer Science, TU Braunschweig, Germany}{s.fekete@tu-bs.de}{https://orcid.org/0000-0002-9062-4241}{}
\author{Vijaykrishna Gurunathan}{Department of Computer Science \& Engineering, IIT Bombay, India}{krishnavijay1999@gmail.com}{https://orcid.org/0000-0001-8580-1269}{}
\author{Kushagra Juneja}{Department of Computer Science \& Engineering, IIT Bombay, India}{kuku12320@gmail.com}{https://orcid.org/0000-0002-4326-6664}{}
\author{Phillip Keldenich}{Department of Computer Science, TU Braunschweig, Germany}{p.keldenich@tu-bs.de}{https://orcid.org/0000-0002-6677-5090}{}
\author{Linda Kleist}{Department of Computer Science, TU Braunschweig, Germany}{l.kleist@tu-bs.de}{https://orcid.org/0000-0002-3786-916X}{}
\author{Christian Scheffer}{Institute of Computer Science, HS Bochum}{christian.scheffer@hs-bochum.de}{https://orcid.org/0000-0002-3471-2706}{}

\authorrunning{S. P. Fekete, K. Juneja, P. Keldenich, L. Kleist, V. Krishna, and C. Scheffer}
\Copyright{S. P. Fekete, K. Juneja, P. Keldenich, L. Kleist, V. Krishna, and C. Scheffer}



\relatedversion{}

\supplement{\url{https://github.com/phillip-keldenich/squares-in-disk}}

\acknowledgements{}

\nolinenumbers 

\hideLIPIcs  

\EventEditors{Kevin Buchin and \'{E}ric Colin de Verdi\`{e}re}
\EventNoEds{2}
\EventLongTitle{37th International Symposium on Computational Geometry (SoCG 2021)}
\EventShortTitle{SoCG 2021}
\EventAcronym{SoCG}
\EventYear{2021}
\EventDate{June 7--11, 2021}
\EventLocation{Buffalo, NY, USA}
\EventLogo{socg-logo.pdf}
\SeriesVolume{189}

\let\BP\undefinedcommand 
\let\SP\undefinedcommand
\newcommand{\TP}{\textsc{Top Packing}\xspace}
\newcommand{\BP}{\textsc{Bottom Packing}\xspace}
\newcommand{\SP}{\textsc{Shelf Packing}\xspace}
\newcommand{\SCP}{\textsc{SubContainer Packing}\xspace}
\newcommand{\Slic}{\textsc{SubContainer Slicing}\xspace}
\newcommand{\RSP}{\textsc{Refined Shelf Packing}\xspace}
\newcommand{\CI}{\textbf{(C\;\!1)}\xspace}
\newcommand{\CII}{\textbf{(C\;\!2)}\xspace}
\newcommand{\CIII}{\textbf{(C\;\!3)}\xspace}
\newcommand{\CTP}{\textbf{(C\;\!3a)}\xspace}
\newcommand{\CBP}{\textbf{(C\;\!3b)}\xspace}
\newcommand{\algo}{\textsc{Layer Packing}\xspace}

\newcommand{\CL}{{C}_\ell\xspace}
\newcommand{\CR}{{C}_r\xspace}

\begin{document}

\begin{titlepage}
	\maketitle
	\thispagestyle{empty}
	\begin{abstract}
	We provide a tight result for a fundamental problem arising from packing squares
	into a circular container: The critical density of packing squares into a disk is
	$\delta=\nicefrac{8}{5\pi}\approx 0.509$. This implies that any set of (not necessarily equal) squares of total area
	$A \leq \nicefrac{8}{5}$ can always be packed into a disk with radius 1; in contrast,
	for any $\varepsilon>0$ there are sets of squares of total area $\nicefrac{8}{5}+\varepsilon$ that cannot be packed,
	even if squares may be rotated.
	This settles the last (and arguably, most elusive) case of packing circular or square objects into a circular or square container:
	The critical densities for squares in a square ($\nicefrac{1}{2}$), circles in a square ($\nicefrac{\pi}{(3+2\sqrt{2})}\approx 0.539)$
	and circles in a circle~($\nicefrac{1}{2}$) have already been established, making use of recursive subdivisions of
	a square container into pieces bounded by straight lines,
	or the ability to use recursive arguments based on similarity of objects and container; 
	neither of these approaches can be applied when packing squares into a circular container.
	Our proof uses a careful manual analysis, complemented by a computer-assisted part that is based on interval arithmetic.
	Beyond the basic mathematical importance, our result is also useful as a blackbox lemma for the analysis of recursive packing algorithms.
	At the same time, our approach showcases the power of a general framework for computer-assisted proofs, based on interval arithmetic.
\end{abstract}
\end{titlepage}

\section{Introduction}
Geometric packing and covering problems arise in a wide range
of natural applications.  They also have a long history of
spawning many demanding (and often still unsolved) mathematical
challenges. These difficulties are also notable from
an algorithmic perspective, as relatively straightforward
one-dimensional variants of packing and covering are already $\NP$-hard \cite{garey1975complexity};
however, deciding whether a given set of one-dimensional segments can be packed
into a given interval can  be checked by computing their total length.
This simple criterion is no longer available for two-dimensional, geometric
packing or covering problems, for which the total area often does not suffice
to decide feasibility of a set, making it necessary to provide an explicit packing or covering. 
A recent result by Abrahamsen et al.~\cite{till2DER} indicates that these difficulties
have far-reaching consequences: Two-dimensional packing problems are $\exists\mathbb{R}$-hard,
so they are unlikely to even belong to \NP.

We provide a provably optimal answer for a natural and previously unsolved 
case of \emph{tight worst-case area bounds}, based on the notion of \emph{critical packing density}:
What is the largest number $\delta\leq 1$, such that any set
of squares with a total area of at most $\delta$ can always be packed (in a not necessarily 
axis-parallel fashion) into a disk $C$ of area 1, regardless of the individual sizes of the squares?
We show the following theorem that implies $\delta=\nicefrac{8}{5\pi}\approx0.509$ for squares in a disk.

\begin{restatable}{theorem}{crit}\label{thm:mainAlgo}
	Every set of squares with a total area  of at most $\nicefrac{8}{5}$ can be packed into the unit disk.
	This is worst-case optimal, i.e., for every $A > \nicefrac{8}{5}$ there exists a set of squares with total area $A$ that cannot be packed into the unit disk.  
\end{restatable}

This critical density $\delta$ is of mathematical importance, as it settles the last open case of 
packing circular or square objects into a circular or square container.
\cref{fig:overview} provides an overview of the critical densities in similar settings, i.e.,
the critical density for packing squares in a square ($\nicefrac{1}{2}$),
disks in a square ($\nicefrac{\pi}{(3+2\sqrt{2})}\approx 0.539$),
and disks in a disk ($\nicefrac{1}{2}$).
\begin{figure}[b!]
	\centering
	\includegraphics[page=2]{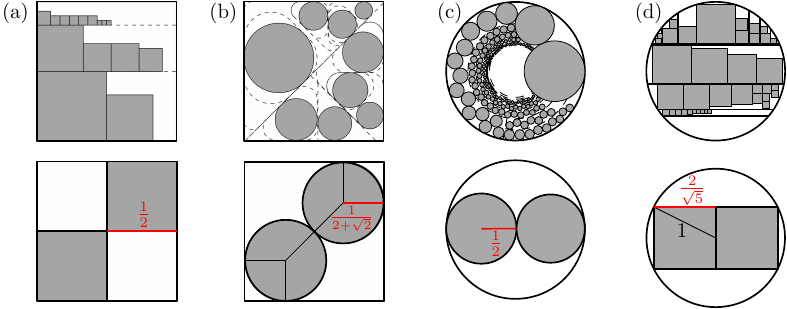}
	\caption{Worst-case optimal approaches and matching worst-case instances for packing: \newline
		(a) Squares into a square with \SP by Moon and Moser~\cite{MM1967some}.\\
		(b) Disks into a square by Morr et al.~\cite{morr2017split,fekete:splitpacking}.\\
		(c) Disks into a disk by Fekete et al.~\cite{fks-pddow-19}.\\
		(d) Squares into a disk [this paper].}
	\label{fig:overview}
\end{figure}

This result is also of algorithmic interest, because it provides a simple sufficient criterion for feasibility. 
Note that the previous results illustrated in \cref{fig:overview} benefitted from recursive subdivisions of a square container into
subpieces bounded by straight lines, or from recursion based on
the similarity of objects and container when both are disks; neither applies when objects are squares
and the container is a disk.
This gives our approach added methodical significance, as it showcases a general framework
for establishing computer-assis\-ted proofs for difficult packing problems for which concise manual
arguments may be elusive.

A proof of \cref{thm:mainAlgo} consists of (i) a class of instances that
provide the upper bound of~$\nicefrac{8}{5\pi}$ for the critical packing density $\delta$
and (ii) an algorithm that achieves the matching lower bound for $\delta$ by packing any 
set of squares with a total area of at
most $\nicefrac{8}{5}$ into the unit disk. The first part is relatively simple:
As shown in \cref{fig:overview}(d), a critical configuration consists of 
two squares of side length $s=\nicefrac{2}{\sqrt 5}$ and a disk~$\mathcal D$ 
of radius 1. It is easy to see that any infinitesimally larger square
(of side length $s+\varepsilon$ for any~$\varepsilon>0$) must contain the
center of $\mathcal D$ in its interior, so two such squares
cannot be packed. \new{Hence, we draw the following conclusion, showing that the bound in \cref{thm:mainAlgo} is tight.}
\new{\begin{lemma}
		For every $\epsilon>0$, there exists a set of squares of total area $\nicefrac{8}{5}+\epsilon$ that cannot be packed into the unit square.
\end{lemma}}

The remainder of our paper focuses on the difficult part:
providing a strategy for packing sets of squares into a disk (described in \cref{sec:algorithm}), 
and then proving that any set of squares with a total area  of at most $\nicefrac{8}{5}$ can indeed be packed
into the unit disk.
This proof is set up with two sets of tools: In \cref{sec:interval-arithmetic-proofs}, we describe 
a general technique that we employed for automated parts of our proof, while 
\cref{sec:analysis-subroutines} provides a number of helpful lemmas.
\cref{sec:analysis-algorithm} gives the actual analysis of our algorithm.

\subsection{Related work: geometric packing}
Problems of geometric packing have been studied for a long time.
Providing a survey that does justice to the wide range of relevant work goes beyond the scope of this paper;
therefore, we strictly focus on very closely related results, in particular, concerning critical packing density.  
We refer to Fejes Tóth~\cite{toth1999recent, toth20172}, Lodi, Martello and Monaci~\cite{LMM2002two},
Brass, Moser and Pach~\cite{brass2006research} and Böröczky~\cite{boroczky2004finite} for more comprehensive surveys.

Even the decision problem whether it is possible to pack a given set of squares into the unit
square was shown to be strongly $\NP$-complete by Leung et
al.~\cite{LTWYC1990packing}, using a reduction from \textsc{3-Partition}.
Already in 1967, Moon and Moser~\cite{MM1967some} proved that it is possible to
pack a set of squares into the unit square if their total area does not exceed
$\nicefrac{1}{2}$.  This bound is best possible, because two squares even
infinitesimally larger than the ones shown in \cref{fig:overview}(a) cannot be
packed. 
The proof is based on a simple recursive argument.

For the scenario with circular objects, Demaine, Fekete, and
Lang~\cite{DFL2010circle} showed in 2010 that deciding whether a given set of
disks can be packed into a unit square is $\NP$-hard.
Using a recursive procedure for partitioning the container into
triangular pieces, Morr, Fekete and
Scheffer~\cite{fekete:splitpacking,morr2017split} proved that the critical
packing density of disks in a square is $\nicefrac{\pi}{(3+2\sqrt{2})} \approx
0.539$.

More recently, Fekete, Keldenich and Scheffer~\cite{fks-pddow-19} established the critical
packing density of disks into a disk. Employing a number of algorithmic
techniques in combination with some interval arithmetic and computer-assisted case
checking, they proved that the critical packing density of disks in a disk is
$\nicefrac{1}{2}$; they also provide a video including an animated overview~\cite{bfk+-pgoow-19}.
In a similar manner, Fekete et al.~\cite{fgk+-wcocrd-20} established a closed-form description of the
total disk area that is sometimes necessary and always sufficient to cover a rectangle depending on its aspect ratio.

Note that the main objective of this line of research is to compute tight worst-case bounds.
For specific instances, a packing may still be possible, even if the density is higher;
this also implies that proofs of infeasibility for specific instances may be trickier.
However, the idea of using the total item volume for computing packing bounds can still
be applied. See the work by Fekete and Schepers~\cite{fs-ncflb-01,fs-gfbhd-04}, which shows how
a {\em modified} volume for geometric objects can be computed, yielding good lower bounds for one- or higher-dimensional
scenarios.

\subsection{Related work: interval arithmetic and computer-assisted proofs}
Establishing tight worst-case bounds for packing problems needs to overcome two main difficulties. 
The first is to deal with the need for accurate computation in the presence of potentially complicated coordinates;
the second is the tremendous size of a full case analysis for a complete proof.

Developing methods for both of these challenges has a long tradition in mathematics.
One of the first instances of interval arithmetic
is Archimedes's classic proof~\cite{pi} that
$\nicefrac{223}{71} \leq \pi \leq \nicefrac{22}{7}$, establishing a narrow interval
for the fundamental constant of geometry.
This entails dealing with inaccurate
computation not merely by giving a close approximation, but by establishing an interval
for the correct value, which can be used for valid lower and upper bounds for subsequent
computations. 

Employing electronic devices (e.g., calculators or computer algebra)
for mathematical arguments is a well-established method for eliminating tedious, error-prone manual calculations.
A famous milestone for the role of computers in theorem proving itself
is the confirmation of the Four Color Theorem, a tantalizing open problem for more than 100 years~\cite{wilson}. 
While the first pair of papers by Appel and Haken~\cite{AH1,AH2}
was still disputed, the universally accepted proof by Robertson et al.~\cite{RSST} still relies on 
extensive use of automated checking. 

Another example is the resolution of the Kepler conjecture by Hales et al.~\cite{kepler_17}: 
While the first version of the proof~\cite{kepler_05}
was still met with some skepticism, the revised and cleaned up
variant~\cite{kepler_17} fits the mold of a more traditional proof, despite
relying both on combinatorial results and computational case checking.
\new{Note that this proof uses a subdivision technique and interval arithmetic in a manner similar to the one used in this paper.
	As in this paper, the result of Hales et al.~\cite{kepler_17} is tight in the numerical sense,
	which means that due to the discretization error introduced by subdividing a space over $\mathbb{R}$ into finitely many pieces,
	parts of the proof must be carried out by other means.}

Other instances of classic geometric problems that were resolved with the help
of computer-assisted proofs are a tight bound for the {E}rdős-{S}zekeres
problem for the existence of convex paths in planar sets of 17
points~\cite{szekeres:computer}, a precursor by Hass and Schlafly~\cite{hass2000double} to the proof of the double bubble theorem by Hutchings et
al.~\cite{hutchings:bubble} (the shape that encloses and separates two given
volumes and has the minimum possible surface area is a standard double bubble,
i.e., three spherical surfaces meeting at angles of $\nicefrac{2\pi}{3}$ on a common
disk), or the proof of $\NP$-hardness of finding a minimum-weight triangulation
(MWT) of a planar point set by Mulzer and Rote~\cite{mulzer:minimum}.

Further examples in the context of
packing and covering include a branch-and-bound approach for covering of
polygons by not necessarily congruent disks with prescribed centers and a
minimal sum of radii by Bánhelyi et al.\!~\cite{Bnhelyi2015OptimalCC}. 

\new{Compared to some more extensive uses of computer assistance in proofs, such as for the four-color theorem,
	in this paper, we use a fairly limited amount of computation to prove our results; we mainly use it to replace 
	an extensive explicit case distinction involving manual analysis and calculus that would be very hard to read
	and verify.}

\section{A worst-case optimal algorithm}
\label{sec:algorithm}
Now we describe \algo,
our worst-case optimal algorithm for packing squares into the unit disk~$\mathcal D$.
The basic idea is to combine refined variants of basic techniques (such as \SP) in several
directions, subdividing the packing area into multiple geometric layers and components.

\subsection{Outline of \algo}
By $s_1,\dots,s_n$, we denote a sequence of squares and simultaneously their side lengths and assume that $s_1 \geq \dots \geq s_n$ is a sorted sequence.
\algo distinguishes three cases that depend on the sizes of the first few, largest squares;
see \cref{fig:Alg} for illustrations.

\begin{description}
	\item[\CI] If $s_1 \leq 0.295$, we place a square of side length $\mathcal{X} = 1.388$ concentric into $\mathcal{D}$
	and place one square of side length  $\mathcal{X}_i= 0.295$, $i\in\{1,\dots,4\}$, to each side of $\mathcal X$, see \cref{fig:Alg}(a).
	The four largest squares $s_1,\ldots,s_4$ are placed in these containers $\mathcal{X}_1,\ldots,\mathcal{X}_4$.
	All other squares are packed into $\mathcal{X}$ using \SP.
	
	\item[\CII] If $s_1 \leq \nicefrac{1}{\sqrt{2}}$ and $s_1^2+s_2^2+s_3^2+s_4^2\geq \nicefrac{39}{25}$,
	let $\mathcal{X}_1,\dots,\mathcal{X}_4$ be four squares of side length $\nicefrac{1}{\sqrt{2}}$ that are placed into $\mathcal{D}$ as depicted in \cref{fig:Alg}(b).
	Furthermore, let $\mathcal{X}$ be a square of side length $\nicefrac{\sqrt{2}}{5}$ that can be packed into~$\mathcal{D}$ in addition to $\mathcal{X}_1,\ldots,\mathcal{X}_4$; see \cref{fig:Alg}(b).
	For $i\leq 4$, $s_i$ is the only square packed into~$\mathcal{X}_i$; this is possible because $s_i\leq s_1\leq \nicefrac{1}{\sqrt{2}}$.
	All other squares are packed into $\mathcal{X}$ using \SP.
	
	\item[\CIII] In the remaining cases, we make extensive use of a refined shelf packing approach.
	Specifically, the largest square $s_1$ is packed into~$\mathcal{D}$ as 
	high as possible, see Figures~\ref{fig:Alg}(c) and \ref{fig:TopBottom}.
	The bottom side of $s_1$ induces a horizontal split of $\mathcal D$ into a \emph{top} and a \emph{bottom} part,
	which are then filled by two subroutines called \TP and \BP, described in \cref{sec:subroutines}.
	For each $i \geq 2$, we then 
	\begin{description}
		\item[\CTP] use \TP to pack $s_i$ if possible,
		\item[\CBP] else we use \BP to pack $s_i$.
	\end{description}
\end{description}

\begin{figure}[h]
	\centering
	\includegraphics[page=3]{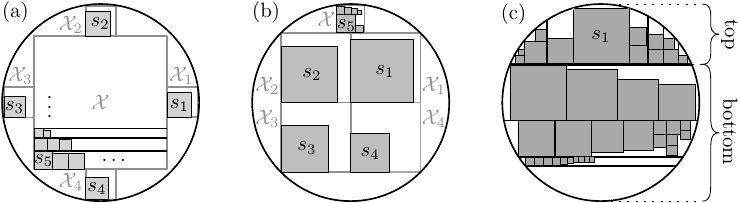}
	\caption{
		Illustration of the packings (a) in Case~\CI,
		(b) in Case~\CII,
		and (c) in Case~\CIII.}
	\label{fig:Alg}
\end{figure}

\begin{figure}[htb]
	\centering
	\includegraphics[page=1]{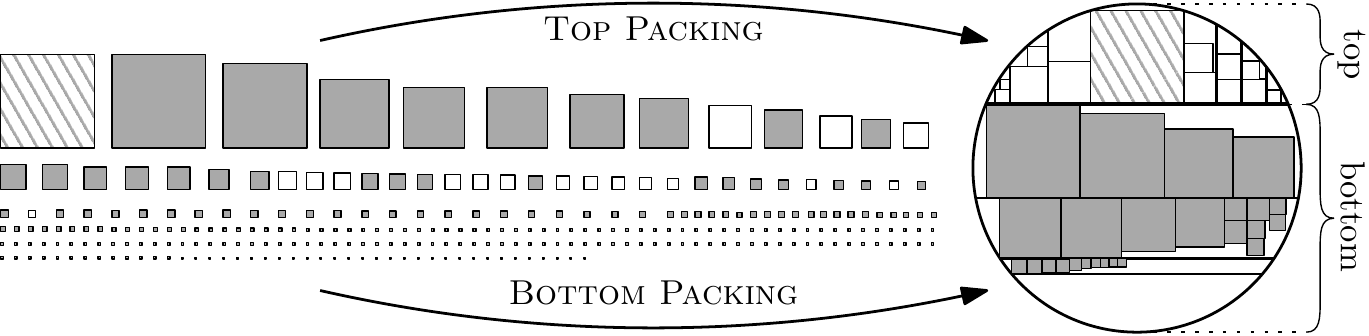}
	\caption{In case \CIII, the largest
		(hatched) square is packed topmost, inducing a top and a
		bottom part of $\mathcal D$.
		Subsequent (white) squares are packed into the pockets of the top part
		with \emph{\TP} (using \RSP as a subroutine) if they fit;
		if they do not fit, they are shown in gray and packed into the bottom part
		with \emph{\BP}, which uses horizontal \Slic, and vertical \RSP
		within each slice.}
	\label{fig:TopBottom}
\end{figure}

\subsection{Subroutines of \algo}
\label{sec:subroutines}
\algo employs a number of different subroutines.
\begin{description}
	\item[\RSP] \label{sec:shelf} The greedy-type packing procedure \SP, employed by Moon and Moser~\cite{MM1967some}, packs objects by decreasing size; see the top of \cref{fig:overview}(a). 
	At each stage, there is a (w.l.o.g.~horizontal) straight cut that separates
	the unused portion of the container from a ``shelf'' into which the next square is packed.
	The height of a shelf is determined by the first packed object. Subsequent
	objects are packed next to each other, until an object no longer
	fits into the current shelf; in this case, a new shelf is opened on top of the previous one.
	For \algo, we use two modifications. 
	
	(1) Parts of the shelf boundaries may be circular arcs; however, we still have a supporting straight axis-parallel boundary and a second, orthogonal straight boundary.
	
	(2) Our refined shelf packing uses the axis-parallel boundary line of a shelf as a support line
	for packing squares; in case of a collision with the circular boundary,
	we may move a square towards the middle of a shelf if this allows packing it. Note that this may only occur in shelves containing the horizontal diameter.
	\item[\TP] The first and largest square $s_1$ is packed as high as possible into $\mathcal D$; see \cref{fig:segmentpacking}(a).
	Then the horizontal line through the bottom of $s_1$ cuts the container into a \emph{top part} that contains $s_1$,
	with two congruent empty pockets $\CL$ and $\CR$ left and right of~$s_1$; and a \emph{bottom part}.
	Each pocket has two straight axis-parallel boundaries, $b_x$ and $b_y$.
	By $\sigma$, we denote the largest square that fits into either pocket.
	For large $s_1$, the bottom side of $\sigma$ does not lie on the same height as the bottom side of $s_1$;
	in that case, we ignore the parts of $\CL$ and $\CR$ that lie below $\sigma$; see \cref{fig:segmentpacking}(e).
	We use \RSP with shelves parallel to the shorter boundary among $b_x$ and $b_y$, 
	as shown in \cref{fig:segmentpacking}(b) and (c).
	If a square does not fit into either pocket, it is packed 
	into the bottom part.
	\begin{figure}[htb]
		\centering
		\includegraphics[page=3]{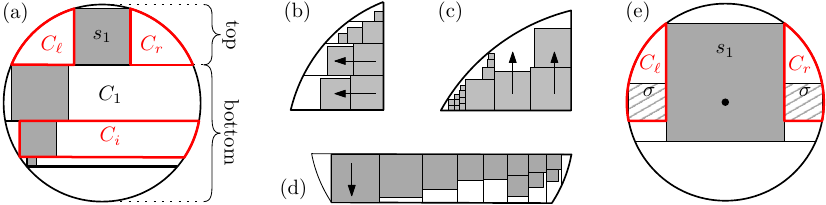}
		\caption{(a) Packing $s_1$ topmost into
			$\mathcal{D}$ yields the top part of $\mathcal{D}$ with  pockets $\CL$ and
			$\CR$, and the bottom part of $\mathcal{D}$. The bottom part is partitioned by \Slic into
			subcontainers~$C_i$, with heights corresponding to the first packed square.\\ 
			(b) A pocket $\CL$ for which $b_x\leq b_y$ implies horizontal shelf
			packing. \\
			(c) A pocket $\CL$ for which $b_x> b_y$ implies vertical shelf
			packing.\\
			(d) Within each subcontainer~$C_i$, \SCP places squares along vertical shelves, starting from the longer straight
			cut of $C_i$.\\
			(e) For large $s_1$, we disregard the parts of $C_{\ell}$ and $C_r$ that lie below their inscribed square $\sigma$.} 
		\label{fig:segmentpacking}
		\label{fig:slicing}
	\end{figure}
	\item [\BP]
	\label{sec:bottom}
	A square that does not fit into the top part of $\mathcal D$ is packed into the bottom part.
	For this purpose, we use (horizontal) \Slic, and (vertical) 
	\SCP within each subcontainer; see \cref{fig:TopBottom} for the overall picture.
	\item[\Slic]
	For packing squares in the bottom part of $\mathcal D$, \Slic subdivides $\mathcal D$ into smaller containers $C_i$,
	by using straight horizontal cuts; see \cref{fig:slicing}(a).
	The height of a subcontainer is determined by the first square packed into it.
	\item [\SCP]
	Within each subcontainer, we use \RSP with vertical shelves.
	These shelves are packed from the longer of the two horizontal cuts, i.e., to pack~$C_i$, we start from  the boundary
	that is closer to the disk center; see \cref{fig:slicing}(d).
\end{description}

\section{Proofs based on interval arithmetic}
\label{sec:interval-arithmetic-proofs}

In interval arithmetic, operations like addition, multiplication or taking the
square root are performed on real intervals $[a,b] \subset \mathbb{R}$ instead
of real numbers.
When applied to intervals, an operation $[a_1,b_1] \circ [a_2,b_2]$ results in the smallest
interval that contains all possible values of $x \circ y$ for $x \in [a_1,b_1], y \in [a_2,b_2]$.
In a practical implementation on computers with finite precision,
computing the smallest such interval is not always possible.
However, using appropriate rounding modes or error bounds, it is still possible to compute an interval that over-approximates the resulting interval,
i.e., contains all possible outcomes of the corresponding real operation.
Predicates such as $[a_1,b_1] \leq [a_2,b_2]$ can also be evaluated on intervals.
The result is a subset of $\{\texttt{false},\texttt{true}\}$ containing all possible outcomes of $x \leq y$ for $x \in [a_1,b_1], y \in [a_2,b_2]$.
This allows evaluating quantifier-free formulas on the Cartesian product of intervals in an over-approximative way.

In many cases throughout this paper, we want to prove that a given non-linear system of real constraints over a bounded $k$-dimensional space $\mathcal{R}$ is unsatisfiable.
This space is typically spanned by a set of $k$ real variables.
Conceptually, to do this in an automatic fashion, we subdivide $\mathcal{R}$ into a sufficiently large number of $k$-dimensional cuboids.
Each such cuboid is defined by an interval for each of the $k$ real variables spanning $\mathcal{R}$.
We then apply interval arithmetic to each such cuboid $\mathcal{C}$ to find a set $S$ of constraints that together eliminate all points of $\mathcal{C}$,
thus proving that no point in $\mathcal{C}$ satisfies all our constraints for a counterexample.

\new{We use this simple technique of subdividing into cubes because it is sufficient for our case, relatively simple to manage and implement and thus arguably less error-prone than relying on deeper, more advanced methods. Moreover, interval arithmetic scales relatively well with the complexity and the number of constraints and can handle non-polynomial constraints involving functions such as $\arccos(x)$ that occur in some of our proofs.}

To improve the efficiency of this approach, our implementation of the basic concept is optimized in several ways.
For instance, the subdivision proceeds in a tree-like fashion according to a fixed ordering of the $k$ variables $v_1,\ldots,v_k$ spanning $\mathcal{R}$.
If variables $v_1,\ldots,v_j$ suffice to exclude a part of~$\mathcal{R}$, we do not split $v_{j+1},\ldots,v_k$ on that part.
Furthermore, we adaptively increase the local fineness of our subdivision if a coarser subdivision does not suffice for some part of $\mathcal{R}$.

Overall, this leads to a limited number of automated proofs\footnote{Source code available at \url{https://github.com/phillip-keldenich/squares-in-disk}.};
for some of these, manual checking would also be feasible, but would involve many case distinctions and would be tedious and unsatisfying.
\revised{Instead, we replace these proofs by the automatic procedure outlined above to have a clear structure instead of an otherwise overwhelming set of arguments.}
\revised{Combined, all automatic proofs required for this paper take less than 1.5 hours and less than 300 MB of memory on the 4 physical cores of the 2.3 GHz Intel i5-8259U CPU in one of the authors' laptops. The proofs involve up to $k = 9$ variables.}

\section{Analysis of subroutines}
\label{sec:analysis-subroutines}
In the following, we establish a number of bounds for the subroutines from \cref{sec:subroutines},
which we use to prove the performance guarantee for \algo.
\new{To prove such a lower bound $\ell$ on the square area packed into some type of container $\mathcal{C}$ by \algo or some subroutine $\mathcal{A}$,
	we typically argue indirectly, by assuming that we have some sequence of squares $s_1,\ldots,s_n$ that $\mathcal{A}$ \emph{fails to pack}.
	W.l.o.g., we always assume that such a sequence is \emph{minimal} in the sense that $\mathcal{A}$ successfully packs $s_1,\ldots,s_{n-1}$,
	but placing $s_n$ according to $\mathcal{A}$ would result in placing a square such that it intersects the exterior of $\mathcal{C}$ or another square $s_i$.
	We then establish our lower bound $\ell$ by showing that the area of $s_1,\ldots,s_n$ must be strictly greater than $\ell$.
	In the remainder of the paper, whenever we say that \algo or a subroutine \emph{fails to pack $s_n$ (or $s_1,\ldots,s_n$)},
	we refer to such a sequence of squares, and always assume that $s_n$ is the first square our algorithm cannot pack.}

\subsection{Shelf Packing}
In several places we make use of the following classic result regarding \SP.
\begin{restatable}[\hspace{-0.2pt}\cite{MM1967some}]{lemma}{shelfpacking}\label{lem:shelfpacking}
	\SP packs every sequence $t_1 \geq \dots \geq t_u$ of squares with a total area of at most $\nicefrac{1}{2}\cdot hw$  into an $h \times w$-rectangle with $t_1 \leq h \leq w$.
\end{restatable}
If the side length of the largest square is small compared to the size of the container, one can guarantee a higher packing density. 
\begin{restatable}{lemma}{shelfpackingtwo}\label{lem:shelfpackingtwo}
	Any finite set of squares with largest square $x_1<\nicefrac{1}{2}$ is packed by
	\SP into a unit square, provided its total area is at most $\nicefrac{1}{2} + 2(x_1 - \nicefrac{1}{2})^2$.
\end{restatable}
\begin{proof}
	We show that if a set of squares cannot be packed, then its area $A$ exceeds the bound of $\nicefrac{1}{2} + 2(x_1 - \nicefrac{1}{2})^2$,
	\new{hence using the same type of indirect argument as outlined in the first paragraph of \cref{sec:analysis-subroutines}.}
	To this end, we assume that the last square in the sequence cannot be packed by \SP, as illustrated in \cref{fig:shelfbound}(a).

	\begin{figure}[htb]
		\centering
		\includegraphics{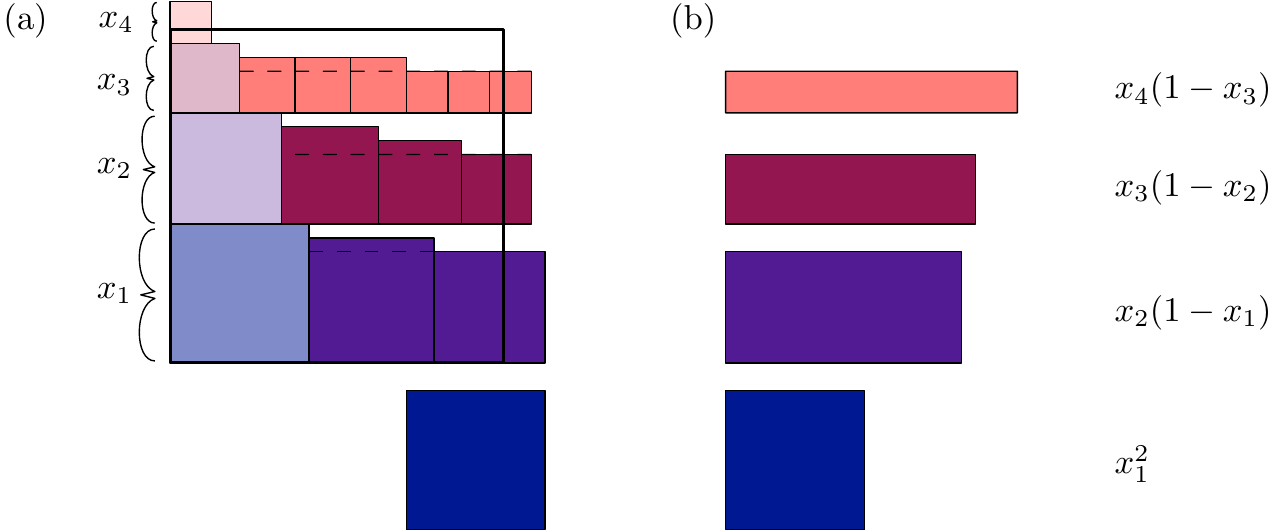}
		\caption{Establishing a refined density bound by shelf packing.
			(a) The area of the first square in each shelf is accounted with the previous shelf.
			(b) Illustration of the lower bound on the area packed in each shelf.}
		\label{fig:shelfbound}
	\end{figure}

	We denote the height of the first square in shelf $i$ by~$x_i$.
	Let $x_1+x_2+\cdots+x_k>1$ be the total height of the arrangement when the last square of the sequence cannot be placed in
	a feasible shelf and is placed in an additional shelf that exceeds the height of the container, i.e., the last square has height $x_k$.
	For computing the area of the packed squares, we account the first square of each shelf with the preceding shelf, see \cref{fig:shelfbound}(a). Then, for $i\geq 1$,  in shelf $i$ the area of the packed squares is at least 
	\[x_{i+1}(1-x_i)\geq x_{i+1}(1-x_1)\] as depicted in illustrated by \cref{fig:shelfbound}(b).
	We can conclude that the total packed area $A$ is
	\[ A > x_1^2 + (x_2+x_3+\dots x_k)(1-x_1) > x_1^2 + (1-x_1)^2 = \nicefrac{1}{2}+ 2(x_1-\nicefrac{1}{2})^2.\hfill\qedhere\]
\end{proof}

\subsection{Top Packing}
\label{sec:anapocketpacking}
We prove the following lower bound on the square area packed by \TP, as long as at least one square fits into the left pocket $\CL$.
By $\sigma = \sigma(s_1)$, we denote the side length of the largest square that can be packed into $\CL$ or $\CR$; see \cref{fig:t-t1-sigmaApp}(c)--(e).
\begin{restatable}{lemma}{anapocketpacking}\label{lem:anapocketpacking} 
	Let $s_1 \geq \cdots \geq s_n$ be a sequence of squares for which \algo fails to pack $s_n$.
	If $s_n \leq \sigma$, then \TP packs squares of total area at least $0.83\sigma^2$.
\end{restatable}

Intuitively, the proof makes use of the \SP bound on the squares inscribed in $\CL$ and $\CR$,
but additionally uses the gaps in $\CR$ and $\CL$ to bound the square area packed into each of $\CL$ and $\CR$ by $0.415\sigma^2$.

Before presenting its proof, we make some helpful observations. We assume the center of our unit disk $\mathcal{D}$ lies at the origin $(0,0)$ of our coordinate system.
Recall that \TP packs the largest square $s_1$ as high as possible into $\mathcal{D}$. This implies that the center of $s_1$ is on the vertical line $x = 0$.
For some $u \in (-1,1)$, we denote by $T(u)$ the side length of the largest square with center on $x = 0$ and bottom on $y = u$ that fits into $\mathcal{D}$; see \cref{fig:t-t1-sigmaApp}(a).

\begin{figure*}[htb]
	\centering
	\resizebox{.99\textwidth}{!}{\includegraphics{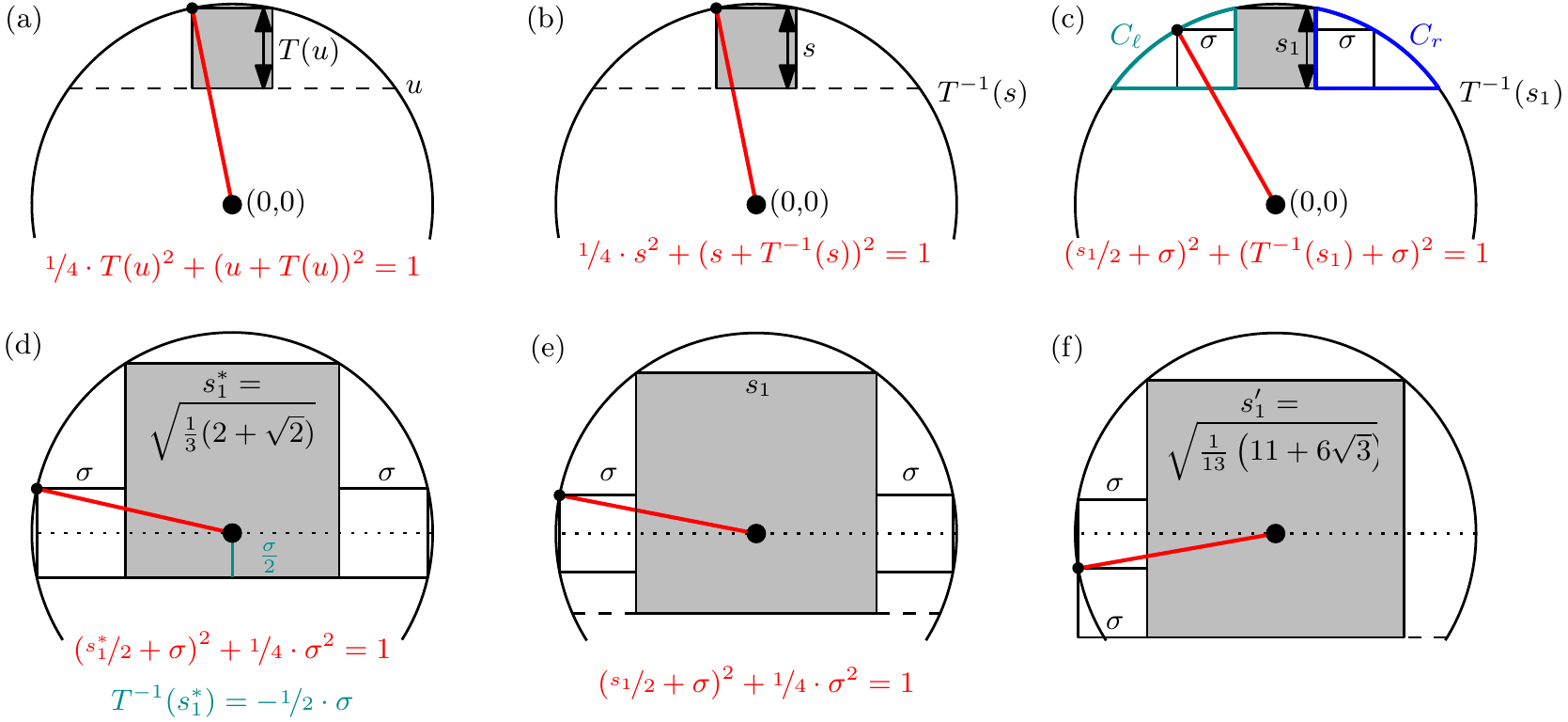}}
	\caption{(a) The definition of $T(u)$ and its defining equation. (b) The definition of $T^{-1}(s)$ and its equation.
		(c) The pockets $\CL$ and $\CR$ used by \TP with their inscribed square $\sigma$ and its equation if $s_1 < s_1^*$,
		(d) $s_1 = s_1^*$ and (e) $s_1 > s_1^*$. (f) The value $s_1' \neq s_1^*$ for $s_1$ for which both cases for $\sigma$ coincide.}
	\label{fig:t-t1-sigmaApp}
\end{figure*}

The inverse function~$T^{-1}(s)$ of $T(u)$ describes the highest possible $y$-coordinate of the bottom side of a square of side length $s$; see \cref{fig:t-t1-sigmaApp}(b). Thus, \TP places the bottom-left corner of $s_1$ at $(-\nicefrac{s_1}{2}, T^{-1}(s_1))$; note that this can be below or above the center of $\mathcal{D}$.
Furthermore, recall that \TP packs the remaining disks into the pockets $\CL$ and $\CR$ induced by placing $s_1$; see \cref{fig:t-t1-sigmaApp}(c).

Now, we present explicit formulas. Solving the equations in \cref{fig:t-t1-sigmaApp}(a)--(b), we get
\[
T(u) = \nicefrac{2}{5} \cdot \left(\sqrt{5-u^2} - 2u\right)\quad \text{ and }\quad 
T^{-1}(s) = \sqrt{1 - \nicefrac{1}{4}\cdot s^2} - s\text{.}\]
To compute $\sigma(s_1)$, we observe the following.
Below some threshold~$s_1^*$, the bottom side of the inscribed square of $\CL$ lies on the horizontal line $y = T^{-1}(s_1)$ and its top left corner touches~$\mathcal{D}$;
see \cref{fig:t-t1-sigmaApp}(c).
\new{In this case, the fact that our disk has radius $1$ implies that $\sigma$ satisfies the equation
	\[\left(\nicefrac{s_1}{2} + \sigma\right)^2 + \left(T^{-1}(s_1) + \sigma\right)^2 = 1,\]
	which has only one positive real solution $\sigma = \nicefrac{1}{4}\cdot\left(-s_1 - 2T^{-1}(s_1) + \sqrt{8 - (s_1 - 2T^{-1}(s_1))^2}\right)$.
	For values $s_1 > s_1^*$ above that threshold, the center of $\CL$'s inscribed square lies on $y = 0$ and both left corners touch the disk; see \cref{fig:t-t1-sigmaApp}(e).
	In this case, $\sigma$ has to satisfy the equation
	\[\left(\nicefrac{s_1}{2} + \sigma\right)^2 + \nicefrac{\sigma^2}{4} = 1,\]
	which has only one positive solution $\sigma = \nicefrac{1}{5}\cdot\left(\sqrt{20 - s_1^2} - 2s_1\right)$.
	There are exactly two positive real values $s_1$ for which these two cases coincide;
	the smaller one yields the threshold value
	$s_1^* = \sqrt{\nicefrac{1}{3}\cdot (2+\sqrt{2})} \approx 1.0668$, see \cref{fig:t-t1-sigmaApp}(d);
	the larger one is $s_1' = \sqrt{\nicefrac{1}{13}\left(11 + 6\sqrt{3}\right)} \approx 1.2828$ and corresponds to the situation depicted in \cref{fig:t-t1-sigmaApp}(f).
	Overall, we obtain the following solution for $\sigma$:}
\begin{equation}\sigma = \begin{cases}
\nicefrac{1}{4}\cdot\left(-s_1 - 2T^{-1}(s_1) + \sqrt{8 - (s_1 - 2T^{-1}(s_1))^2}\right), & \text{if }s_1 \leq s_1^*,\\
\nicefrac{1}{5}\cdot\left(\sqrt{20 - s_1^2} - 2s_1\right), & \text{otherwise.}
\end{cases}
\label{eq:sigma}
\end{equation}
Now we are ready to present a proof of \cref{lem:anapocketpacking}.

\begin{proof}[Proof of \cref{lem:anapocketpacking}]
	We begin by observing that $s_n \leq \sigma$ would fit into either $\CL$ or $\CR$; as we fail to pack $s_n$, $\CL$ and $\CR$ must contain other squares.
	In the following, we prove that \TP packs squares of area $A \geq 0.415\sigma^2$ into $\CL$.
	An analogous argument works for~$\CR$, implying an overall bound of $0.83\sigma^2$.

	By $\ell_1 \coloneqq \sqrt{1 - T^{-1}(s_1)^2} - \nicefrac{s_1}{2}$, we denote the length of the bottom boundary of $\CL$; see Figure~\ref{fig:lemma_interval1}(a).
	W.l.o.g., we assume $\ell_1 \leq s_1$; the other case is symmetric.
	In other words, we assume that the bottom boundary of $\CL$ is shorter than its right boundary,
	which means that we are using horizontal shelves that we fill from right to left as depicted in \cref{fig:segmentpacking}(b).
	
	\begin{figure*}[b!]
		\centering
		\includegraphics{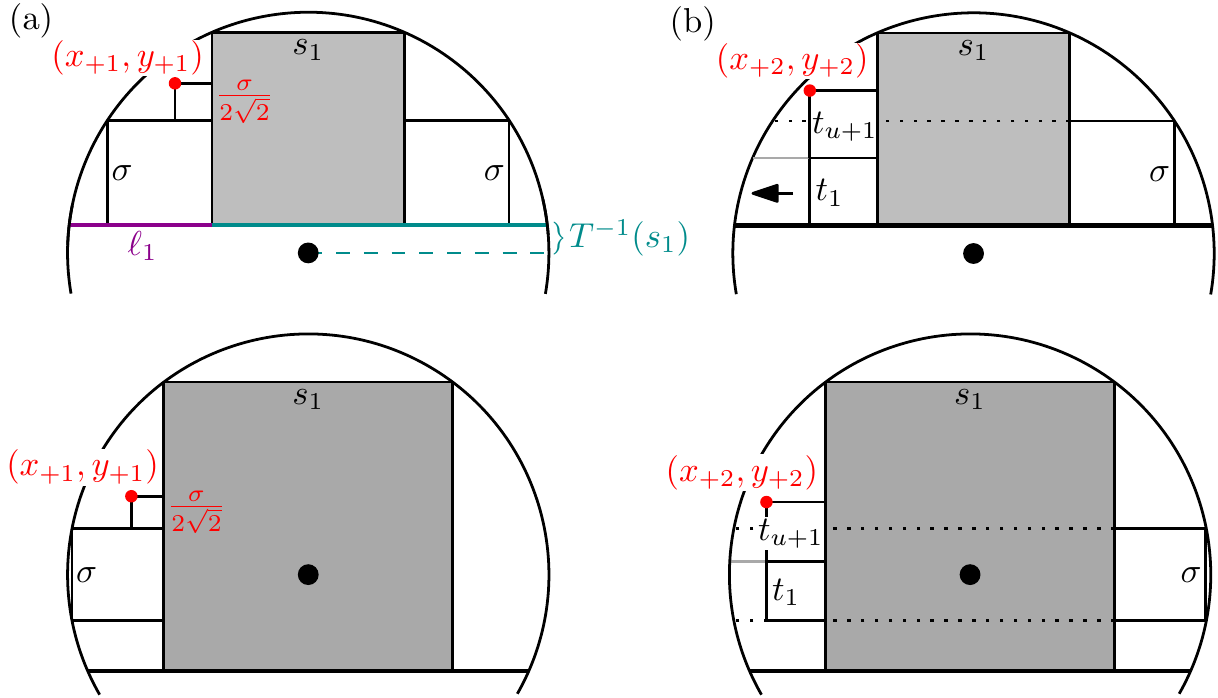}
		\caption{Illustration of the proof of \cref{lem:summaryInterval1}. In both parts, the red point is always contained in the disk $\mathcal{D}$. (a) The values occurring in part (1).
			(b) The situation for $t_1, t_{u+1}$ of maximum possible size in part (2).}
		\label{fig:lemma_interval1}
	\end{figure*}
	
	Consider the subsequence $t_1,\ldots,t_u,t_{u+1}$ of $s_1,\ldots,s_n$, where $t_1,\ldots,t_u$ are the squares packed by \TP into $\CL$ before height $\sigma$ is (strictly) exceeded,
	and $t_{u+1}$ is the next square that we try to pack into $\CL$.
	We observe that $t_{u+1}$ may or may not be packed into $\CL$ by \TP, and that $u \geq 1$ by $t_1 \leq \sigma$, i.e., after placing the first square, height $\sigma$ is not exceeded.
	We make use of the following lemma, proved by interval arithmetic.
	\begin{lemma}[Automatic Analysis for \TP]\label{lem:summaryInterval1}
		Let $\ell_1 \leq s_1$, $x_{+1} = \nicefrac{s_1}{2} + \nicefrac{\sigma}{2\sqrt{2}}$ and 
		$x_{+2} = \nicefrac{s_1}{2} + 0.645\sigma$.
		Furthermore, let
		\[ y_{+1} = \begin{cases}
		T^{-1}(s_1) + \sigma + \nicefrac{\sigma}{2\sqrt{2}}, &\text{if }s_1 \leq s_1^*,\\
		\nicefrac{\sigma}{2} + \nicefrac{\sigma}{2\sqrt{2}}, &\text{otherwise,}
		\end{cases}\quad
		y_{+2} = \begin{cases}
		T^{-1}(s_1) + 2 \cdot 0.645\sigma, &\text{if }s_1 \leq s_1^*,\\
		-\nicefrac{\sigma}{2} + 2 \cdot 0.645\sigma, &\text{otherwise;}
		\end{cases}
		\]
		see \cref{fig:lemma_interval1}. Let \(F_{TP_1}(s_1) \coloneqq x_{+1}^2 + y_{+1}^2\) and \(F_{TP_2}(s_1) \coloneqq x_{+2}^2 + y_{+2}^2\).
		Then, for all \(0.295 \leq s_1 \leq \sqrt{\nicefrac{8}{5}}\), we have \textup{(1)} $F_{TP_1}(s_1) \leq 1$ and \textup{(2)} $F_{TP_2}(s_1) \leq 1$.
	\end{lemma}
	
	If $0.645\sigma \leq t_1$, the packed area inside $\CL$ is at least $t_1^2 \geq 0.645^2\sigma^2 > 0.415\sigma^2$.
	Thus, in the following, we assume $t_1 < 0.645\sigma$.
	
	Furthermore, if $t_{u+1} \leq \nicefrac{\sigma}{2 \sqrt{2}}$, we can apply \cref{lem:summaryInterval1}~(1), showing that $t_{u+1}$ can be packed into $\CL$ by \TP; see Figure~\ref{fig:lemma_interval1}(a).
	In particular, $t_{u+1}$ can always be packed into $\CL$ such that its bottom side lies on height $\sigma$ and its right side touches $s_1$.
	The total area packed by \TP into $\CL$ is at least the total area packed by \SP into the square of area $\sigma$; here we use the fact that the height of the bottom segment of a pocket and the bottom segment of the contained square $\sigma$ coincide, see also \cref{fig:segmentpacking}(e).
	Because packing $t_{u+1}$ exceeds height $\sigma$, \cref{lem:shelfpackingtwo} implies that the total area of $t_{u+1}$ and the squares already packed into $\CL$ exceeds $\nicefrac{\sigma^2}{2}$. 
	Thus, in the following, we assume $\nicefrac{\sigma}{2 \sqrt{2}} < t_{u+1} \leq t_1 < 0.645\sigma$.
	
	If $t_1 \leq \nicefrac{\sigma}{2}$, at least four squares are packed into $\CL$ by \RSP before height $\sigma$ is exceeded.
	Consequently, the total packed area is at least $4\left( \nicefrac{\sigma}{2 \sqrt{2}} \right)^2 = \nicefrac{\sigma^2}{2}$. 
	Thus, in the following we assume $t_1 > \nicefrac{\sigma}{2}$.
	
	Now let us assume that only one shelf is constructed before height $\sigma$ is exceeded.
	That shelf has height $t_1$ and thus we must have $t_{u+1} > \sigma - t_1$.
	We use \cref{lem:summaryInterval1}~(2) to prove that we can pack $t_{u+1}$ on top of the first shelf, even when assuming that $t_1 = t_{u+1} = 0.645\sigma$ are as large as possible; see Figure~\ref{fig:lemma_interval1}(b).
	Thus, the total area packed into $\CL$ is at least $t_1^2 + t_{u+1}^2 \geq t_1^2 + (\sigma - t_1)^2 \geq \nicefrac{\sigma^2}{2}$.
	
	Otherwise, at least two shelves are constructed before height $\sigma$ is exceeded.
	The first shelf has height $t_1$.
	The second shelf contains at least two squares because its height is at most $\sigma - t_1 \leq \nicefrac{\sigma}{2}$, and thus at most half of its width.
	Thus, the area packed into $\CL$ is at least $t_1^2 + 2t_{u+1}^2 \geq \nicefrac{\sigma^2}{4} + 2(\nicefrac{\sigma}{2\sqrt{2}})^2 = \nicefrac{\sigma^2}{2}$, concluding the proof of \cref{lem:anapocketpacking}.
\end{proof}

\subsection{Subcontainer Packing}

For the analysis of \SCP, let $C_1,\dots,C_k$ be the subcontainers constructed by \BP and let $R_1,\dots,R_k$ be the maximal rectangles contained in $C_1,\dots,C_k$; see \cref{fig:scp-overview}.

\begin{figure}[ht]
	\centering
	\begin{minipage}[t]{.29\textwidth}
		\centering
		\includegraphics[scale=1]{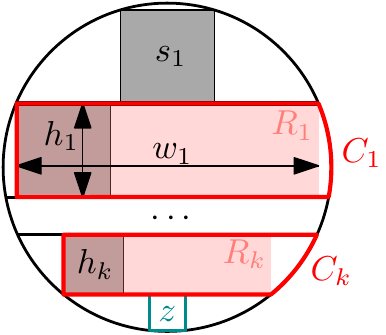}
		\caption{Subcontainers $C_i$, $ i\leq k$, produced by \Slic.
		}
		\label{fig:scp-overview}
	\end{minipage}\hfill
	\begin{minipage}[t]{.67\textwidth}
		\centering
		\includegraphics[scale=1]{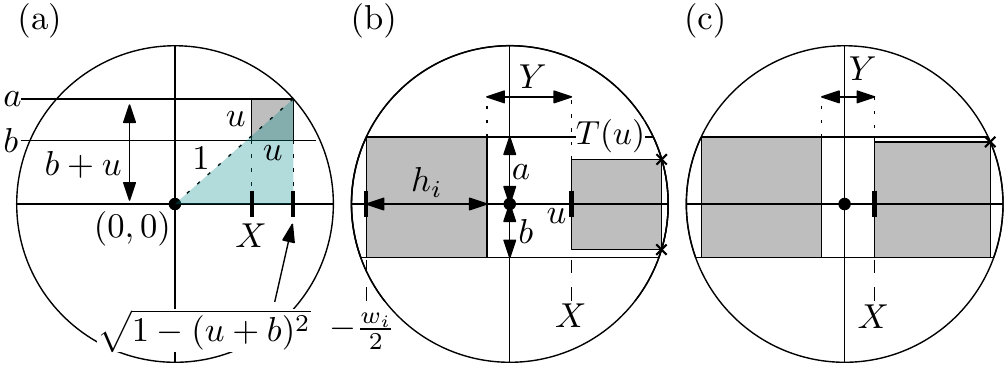}
		\caption{The computation of $X$ for a square of side length~$u$:
			(a) in case $b \geq 0$, which is symmetric to $a < 0$;
			(b) in case $a > 0 > b$ and $u \leq 2c$;
			(c) in case $a > 0 > b$ and $u > 2c$.}
		\label{fig:functionG}
	\end{minipage}
\end{figure}

For $i = 1,\dots,k$, let $h_i$ and $w_i$ denote the height and the width of $R_i$.
Recall that $h_i$ simultaneously denotes the height of $C_i$ and the first square packed into $C_i$.
Let $z$ be the largest square that could be packed below $C_k$.
We define $h_{k+1} \coloneqq s_n$, so $h_{i+1}$ always denotes the first square that did not fit into $C_i$.
Furthermore, we denote the total area of squares packed into $C_i$ by $\| C_i \|$.
We establish several lower bounds on this area $\| C_i \|$.
One such bound is derived from the following observation.

\begin{observation}\label{lem:betterShelfPacking}
	The total area packed by \SCP into $C_i$ is at least the total area packed by \SP into $R_i$.
\end{observation}

If the width of $R_i$ is at least twice its height, the following lemma improves on this bound. 
\begin{restatable}{lemma}{Aone}\label{cor:Aone}
	For every sequence of squares $s_1,\ldots,s_n$ for which \algo constructs at least $i$ subcontainers and fails to pack $s_n$, if $w_i \geq 2 h_i$, the area packed into~$C_i$ is
	\begin{align*}
	\| C_i \| \geq B_1(h_i,w_i,h_{i+1})
	\coloneqq & \max \begin{cases}
	\nicefrac{1}{2}\cdot h_i w_i + \nicefrac{1}{4}\cdot h_i^2,\\
	h_i^2 + (w_i - h_i - h_{i+1})h_{i+1},\\
	\nicefrac{1}{2}\cdot h_i (w_i + h_i)  - h_{i+1}^2.
	\end{cases}
	\end{align*}
\end{restatable}
For better readability, we present the proof of \cref{cor:Aone} in the end of this subsection, namely in  \cref{app:scp}.

Moreover, we can extend $B_1$ to the cases where the width of $R_i$ is smaller than twice its height as follows. 
\begin{restatable}{lemma}{Aeleven}\label{lem:Aeleven}
	For every sequence of squares $s_1,\ldots,s_n$ for which \algo constructs \new{exactly $j \geq i$ subcontainers and fails to pack $h_{j+1} := s_n$, 
		the area  packed into~$C_i$ is}
	\begin{align*}
	\| C_i \| &\geq B_{2}(h_i,w_i,h_{i+1})
	:=	\begin{cases}
	h_i^2 & \text{ if } w_i < h_i + h_{i+1}, \phantom{\leq w_i \leq  2 h_{i}}\\
	h_i^2 + h_{i+1}^2 & \text{ if }\;\; \phantom{w_i \leq} h_i + h_{i+1} \leq w_i < 2 h_{i}, \\
	B_1(h_i,w_i,h_{i+1}) & \text{ if }\; \phantom{w_i \leq h_i + h_{i+1} \leq w_i \leq \ } 2 h_{i} \leq w_i.
	\end{cases}
	\end{align*}
\end{restatable}
\begin{proof}
	We always pack at least the square $h_i$ into $C_i$.
	As $h_{i+1} \leq h_i$, if $h_i + h_{i+1} \leq w_i$, we pack at least two squares into $C_i$:
	Let $s_j$ be the second square we consider packing into $C_i$.
	If $h_i$ and $s_j$ do not fit into $C_i$, then $s_j = h_{i+1}$, as we would open a new subcontainer for $s_j$.
	This contradicts $h_i + h_{i+1} \leq w_i$, as $h_i$ and $h_{i+1}$ fit into $R_i$ and thus into $C_i$.
\end{proof}

For the last lemma of this subsection, we introduce some useful notation.
Let $a,b$ be the $y$-coordinates of the upper and the lower side of~$C_i$ and let $c\coloneqq c(a,b) = \min \{ a,-b \}$. 
When $C_i$ contains the center of the disk, i.e.,  $a> 0>b$, $c$ denotes the distance of the origin to the nearer side of $C_i$, see \cref{fig:functionG}(b) and (c).
The maximal $x$-coordinate $X$ of the left side of a square of side length $u$ in $C_i$ is determined by
\begin{align*}
X(a,b,u) \coloneqq& 
\left.
\begin{cases}
\sqrt{1-(u+b)^2} - u & \text{if } b\geq 0,\\
\sqrt{1-(u-a)^2} - u & \text{if } a < 0,\\
T^{-1}(u) & \text{else if } u \leq 2 c,\\
\sqrt{1-(u-c)^2} - u & \text{otherwise} 
\end{cases}
\right\rbrace
=\begin{cases}
T^{-1}(u) & \text{if } u \leq 2 c,\\
\sqrt{1-(u-c)^2} - u & \text{otherwise}
\end{cases}
\end{align*}

The $x$-coordinate of the right side of the first square $h_i$ packed into subcontainer~$C_i$ is $-\nicefrac{1}{2}\cdot w_i + h_i$.
Thus, as $h_{i+1}$ did not fit into $C_i$, we can lower bound the total width of squares packed into $C_i$ after~$h_i$, see \cref{fig:lowerBoundA2}(a), by
\[Y \coloneqq Y(a,h_i,w_i,h_{i+1}) \coloneqq \nicefrac{1}{2} \cdot w_{i} - h_{i} + X(a,a-h_i,h_{i+1}).\]

\begin{figure*}[b]
	\centering
	\includegraphics[page=3]{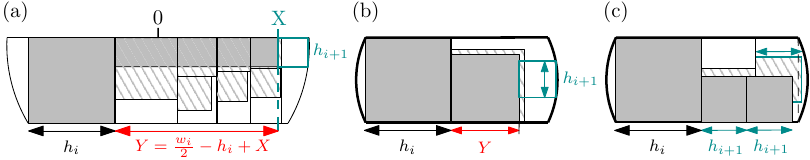}
	\caption{(a) Definition of $Y$.
		(b) The lower bound $B_3(a, h_i,w_i,h_{i+1})$ when two squares are packed into $C_i$.
		(c) The lower bound $B_3(a, h_i,w_i,h_{i+1})$ when at least three squares packed into~$C_i$.}
	\label{fig:lowerBoundA2}
\end{figure*}

\begin{restatable}{lemma}{AtwoBound}\label{lem:AtwoBound}
	For every sequence of squares $s_1,\ldots,s_n$ for which \algo constructs exactly $j \geq i$ subcontainers and fails to pack $h_{j+1} := s_n$,
	the area packed into~$C_i$ is at least
	\[ \| C_i \| \geq B_3(a, h_i, w_i, h_{i+1}) \coloneqq \max \begin{cases}
	h_i^2 + \max\{0,Y(a,h_i,w_i,h_{i+1})\} \cdot h_{i+1}, & \textup{(10.1)}{}\\ 
	h_i^2 + \min\left(\max^2(Y(a,h_i,w_i,h_{i+1}),0), 2h_{i+1}^2\right). & \textup{(10.2)}
	\end{cases}\]
\end{restatable}
\begin{proof}
	If $Y \leq 0$, both bounds (10.1) and (10.2) simplify to $h_i^2$ and are valid, because $h_i$ is packed into~$C_i$.
	Thus, let us assume $Y > 0$.
	This implies that there are at least two squares packed into $C_i$; otherwise, $h_{i+1}$ would  fit into $C_i$.
	As $Y$ is a lower bound on the total width of squares packed into~$C_i$ after $h_i$ and $h_{i+1}$ is a bound on their height,
	we obtain bound~(10.1); see \cref{fig:lowerBoundA2}(a).
	
	Furthermore, if exactly two squares are packed into $C_i$, $Y^2$ can be used as lower bound on the area of the second square, see \cref{fig:lowerBoundA2}(b).
	Otherwise, at least three squares are packed into~$C_i$, and we can add $2h_{i+1}^2$ to $h_i^2$ to bound their area, see \cref{fig:lowerBoundA2}(c).
	Combining these two cases yields bound~(10.2).
\end{proof}
We combine these previous bounds into a general lower bound for $\| C_i \|$.
\begin{restatable}{corollary}{combined}\label{cor:combined}
	For every sequence of squares $s_1,\ldots,s_n$ for which \algo constructs exactly $j \geq i$ subcontainers and fails to pack $h_{j+1} := s_n$, 
	the area packed into~$C_i$ is bounded by
	\[\| C_i \| \geq B_4(a, h_i,w_i,h_{i+1}) := \max \begin{cases}
	B_2(h_i, w_i, h_{i+1}), & \text{(\cref{lem:Aeleven})}\\
	B_3(a, h_i, w_i, h_{i+1}). & \text{(\cref{lem:AtwoBound})}
	\end{cases}\]
\end{restatable}

\subsubsection{Proof of \cref{cor:Aone}}\label{app:scp}
In this \new{sub}section, we provide the proof of \cref{cor:Aone}. 
\Aone*

\new{The proof consists of the three \cref{lem:shelfpackingnew,lem:anaFramesOfBasePacking,lem:anaFramesOfBasePackingTwo}, each establishing the validity of one of the three lower bounds combined in $B_1$.
	In the setting of \cref{cor:Aone}, we are dealing with a sequence of squares $s_1,\ldots,s_n$, for which \algo fails to pack $s_n$. 
	Furthermore, \algo constructs subcontainer $C_i$; thus, there must be a square $t_1 = h_i$ that is placed into $C_i$ (implying that $h_i\leq w_i$), and a subsequence $t_1,\ldots,t_{u},t_{u+1}$ of squares that \SCP tries to pack into $C_i$,
	where $t_u$ denotes the last square that is packed into $C_i$, and $t_{u+1}$ does not fit.}

\begin{restatable}{lemma}{shelfpackingnew}\label{lem:shelfpackingnew} 
	{Let $t_1 = h_i \leq w_i$. If $w_i\geq 2h_i$
		then the total area of the squares $t_1,\dots,t_u$ that \SCP packs into $C_i$ is  $\|C_i\|\geq \nicefrac{1}{2}\cdot h_iw_i + \nicefrac{1}{4}\cdot h_i^2$.}
\end{restatable}
\begin{proof}
	If $t_{u+1} \leq \nicefrac{1}{2}\cdot h_i$, we consider the part~$R$ of the rectangle $R_i$ that remains after removing $t_1$.
	Because $t_{u+1}$ is not packed by \SCP into $C_i$, \cref{lem:shelfpacking,lem:betterShelfPacking} imply that the total area of
	$t_2, \dots, t_{u+1}$ is at least $\nicefrac{1}{2}\cdot(w_i-h_i)h_i$.
	Consequently, the total area of $t_1,\ldots,t_u$ is at least \[h_i^2 + \nicefrac{1}{2}(w_i-h_i)h_i - t_{u+1}^2 \geq \nicefrac{1}{2} \cdot h_iw_i+ \nicefrac{1}{2}\cdot h_i^2 - \nicefrac{1}{4}\cdot h_i^2 = \nicefrac{1}{2}\cdot h_iw_i+ \nicefrac{1}{4}\cdot h_i^2.\]
	
	Thus, we may restrict ourselves to $t_{u+1} =\nicefrac{1}{2}\cdot h_i + \delta$ for some $\delta > 0$; see \cref{fig:area-bound-overview}.
	\begin{figure}[htb]
		\centering
		\includegraphics{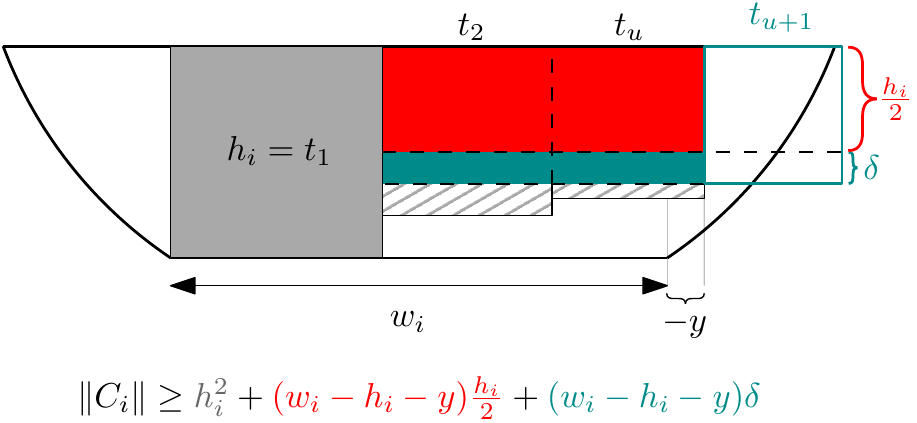}
		\caption{Our bound on the area $\|C_i\|$ in the remaining case.}
		\label{fig:area-bound-overview}
	\end{figure}
	Let $y \coloneqq w - t_1 - \ldots - t_u$; note that $y$ may be negative or positive, but we have $y \leq t_{u+1} = \nicefrac{1}{2}\cdot h_i + \delta$, as otherwise, we could have packed $t_{u+1}$.
	We can bound the area packed into $C_i$ by 
	\begin{align*}
	\|C_i\| &\geq h_i^2 + (w_i - h_i - y)\frac{h_i}{2} + (w_i - h_i - y)\delta\\
	&= \frac{1}{2}h_i^2 + \frac{1}{2}w_ih_i - \frac{1}{2}yh_i + (w_i - h_i - y)\delta\\
	&= \frac{1}{2} w_ih_i+ \frac{1}{4}h_i^2 + \underbrace{\frac{1}{4} h_i^2+ - \frac{1}{2}yh_i + (w_i - h_i - y)\delta}_{\text{(I)}}.
	\end{align*}
	
	Therefore, it suffices to prove that $\text{(I)} = \frac{1}{4}h_i^2  - \frac{1}{2}yh_i + (w_i - h_i - y)\delta \geq 0$, \new{which is equivalent to} $ (w_i - h_i - y) \delta \geq \frac{1}{2}h_i \left(y - \frac{1}{2}h_i \right)$.
	We observe that both factors on the left \new{hand} side and one factor on the right \new{hand} side are non-negative, and distinguish whether $w_i \geq 2 h_i + \delta$ or not.
	If $w_i \geq 2 h_i + \delta$, from $y \leq t_{u+1}$ we obtain $w_i - h_i - y \geq \frac{1}{2}h_i$ and $\delta \geq y - \frac{1}{2}h_i$, finishing this case.
	
	Finally, we have to handle the case $2h_i \leq w_i \leq 2h_i + \delta$, which implies that only two squares are packed into $C_i$, i.e., $u = 2$.
	The total area of $t_1,t_2$ is lower bounded by 
	\begin{align*}
	t_1^2 + t_2^2 &\geq h_i^2 + t_3^2 = h_i^2 + \left(\frac{1}{2}h_i + \delta \right)^2\\
	&= h_i^2 + \frac{1}{4}h_i^2 + h_i\delta + \delta^2 \geq \frac{1}{4}h_i^2 + h_i^2 + \frac{1}{2}h_i\delta\\
	&= \frac{1}{4}h_i^2 + \underbrace{(2h_i + \delta)}_{\geq w_i}\frac{1}{2}h_i \geq \frac{1}{2}h_iw_i + \frac{1}{4}h_i^2,
	\end{align*}
	concluding the proof.
\end{proof}

\new{Now, we establish the second lower bound used in \cref{cor:Aone}.}
\begin{restatable}{lemma}{anaFramesOfBasePacking}\label{lem:anaFramesOfBasePacking}
	\new{	For every sequence of squares $s_1,\ldots,s_n$ for which \algo constructs exactly $j$ subcontainers and fails to pack $s_n$,
		the total area of squares that \SCP packs into $C_i$ for all $i\leq j$  is $\|C_i\| \geq h_i^2 + h_{i+1}(w_i - h_i - h_{i+1})$,
		with $h_{j+1} \coloneqq s_n$.}
\end{restatable}

\begin{proof}
	The square that failed to be packed into $R_i$ has a side length of $h_{i+1}$, see \cref{fig:LowerBoundStrippacking}.
	W.l.o.g., we assume that we fill the vertical shelves in $R_i$ from the top.
	The width of the piece of the top of $R_i$ that is not covered by squares is less than $h_{i+1}$;
	otherwise, $h_{i+1}$ would have fit.
	Therefore, the total square area packed into $C_i$ after $h_i$ is at least $h_{i+1}(w_i - h_i - h_{i+1})$,
	because the height of all squares packed into $C_i$ is at least $h_{i+1}$, concluding the proof.
\end{proof}
\begin{figure}[htb]
	\centering
	\includegraphics{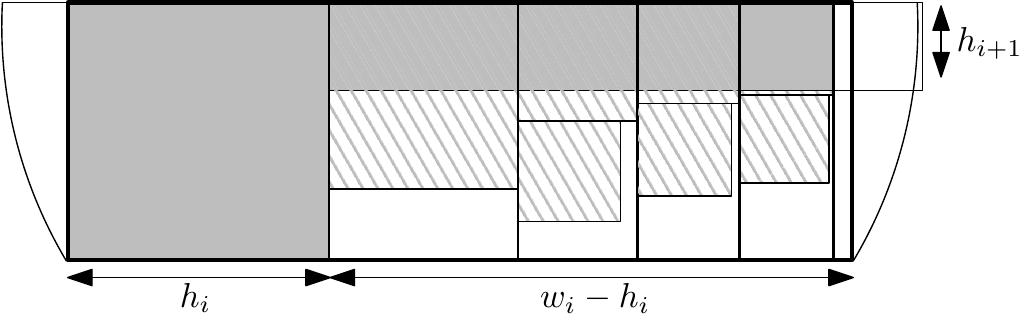}
	\caption{Illustration for the proof of \cref{lem:anaFramesOfBasePacking}:  A lower bound (gray) for the area (hatched) packed by \SCP into a subcontainer.}
	\label{fig:LowerBoundStrippacking}
\end{figure}

\new{Finally, we establish the third lower bound used in \cref{cor:Aone}.}

\begin{restatable}{lemma}{anaFramesOfBasePackingTwo}\label{lem:anaFramesOfBasePackingTwo}
	\new{For every sequence of squares $s_1,\ldots,s_n$ for which \algo constructs exactly $j$ subcontainers and fails to pack $s_n$, if $w_i \geq 2h_i$,
		the total area of squares that \SCP packs into $C_i$ for all $i\leq j$ is $\|C_i\| \geq \nicefrac{1}{2}\cdot h_i(w_i+h_i) - h_{i+1}^2$, with $h_{j+1} \coloneqq s_n$.}
\end{restatable}
\begin{proof}
	We consider the remainder $R_i'$ of rectangle $R_i$ after $h_i$ is packed.
	$R_i'$ has height $h_i$ and width $w_i-h_i$. \new{Because $w_i\geq2h_i$ we have that $w_i-h_i \geq h_i$ and hence, at least one square is packed into $R_i'$.}
	Consequently, \cref{lem:shelfpacking,lem:betterShelfPacking} imply that the total area packed by \SCP into $C_i$ is at least $\nicefrac{1}{2}\cdot h_i(w_i - h_i)-h_{i+1}^2$.
	Thus, we have $\|C_i\| \geq h_i^2 + \nicefrac{1}{2}\cdot h_i(w_i - h_i) -h_{i+1}^2= \nicefrac{1}{2}\cdot h_i (w_i +h_i) - h_{i+1}^2$. 
\end{proof}

\section{Analysis of the main algorithm}
\label{sec:analysis-algorithm}
In this section, we prove our main result using the tools provided in \cref{sec:interval-arithmetic-proofs,sec:analysis-subroutines}.
On the highest level, the proof consists of three parts corresponding to the three cases that our algorithm distinguishes.

\subsection{Analysis of \CI}
\label{sec:analysis-first-case}
Recall that in case \CI, we place a container square $\mathcal{X}$ of side length $1.388$ into $\mathcal{D}$,
and pack the first four squares into pockets outside $\mathcal{X}$ and all remaining disks into $\mathcal{X}$ using \SP; see \cref{fig:Alg}(a).

\begin{restatable}{lemma}{correctnessStepOne}\label{lem:correctnessStepOne}
	If \algo fails to pack a sequence of squares $s_1,\ldots,s_n$ with $s_1 \leq 0.295$, the total area of the squares exceeds $\nicefrac{8}{5}$.
\end{restatable}
\begin{proof}
	Consider scaling down all side lengths by a factor of $\nicefrac{1}{1.388}$, such that $\mathcal X$ is the unit square and $s_1 \leq \nicefrac{0.295}{1.388} \approx 0.2125$.
	As $s_5$ is the first square packed by \SP into $\mathcal X$,
	\cref{lem:shelfpackingtwo} implies that the total area packed into the scaled $\mathcal{D}$ is at least
	$f(s_5) = 4 s_5^2 + \nicefrac{1}{2} + 2(s_5-\nicefrac{1}{2})^2$ with derivative $f'(s_5)= 12s_5 - 2,$
	which is minimized for $s_5 = \nicefrac{1}{6}$, where $f(\nicefrac{1}{6}) = \nicefrac{5}{6}$.
	Thus, in the non-scaled configuration, the area packed is at least
	$\nicefrac{5}{6}\cdot 1.388^2 = \nicefrac{120409}{75000} \approx 1.605 > \nicefrac{8}{5}$, concluding the proof.
\end{proof}

\subsection{Analysis of \CII}
\label{sec:analysis-second-case}
Recall that in case \CII, we pack the four largest squares into the squares  $\mathcal{X}_1,\dots, \mathcal{X}_4$;
all other squares are packed into a square container $\mathcal{X}$ on top of them; see \cref{fig:Alg}(b).
\begin{restatable}{lemma}{correctnessStepTwo}\label{lem:correctnessStepTwo}
	If \algo fails to pack a sequence $s_1,\ldots,s_n$ of squares with $0.295 < s_1 \leq \nicefrac{1}{\sqrt{2}}$ and $s_1^2 + s_2^2 + s_3^2 + s_4^2 \geq \nicefrac{39}{25}$, 
	the total area of the squares exceeds $\nicefrac{8}{5}$.
\end{restatable}
\begin{proof}
	By assumption, the total area of the squares $s_1,s_2,s_3,s_4$ is at least $\nicefrac{39}{25} = \nicefrac{8}{5} - \nicefrac{1}{25}$.
	As $\mathcal{X}$ has an area of $\nicefrac{2}{25}$, \cref{lem:shelfpacking} implies that \SP (and thus \algo) only fails to pack all remaining squares into $\mathcal{X}$ if their area exceeds $\nicefrac{1}{25}$. Consequently, the total area of the squares exceeds $\nicefrac{8}{5}$, concluding the proof. 
\end{proof}

\subsection{Analysis of \CIII}
Recall that $z$ denotes the largest square that could be packed below the last subcontainer $C_k$ constructed by \BP,
as illustrated in \cref{fig:scp-overview} or in \cref{fig:formula-z}, where we have reflected the instance along the x-axis.
\begin{figure*}[b]
	\centering
	\includegraphics[scale=1]{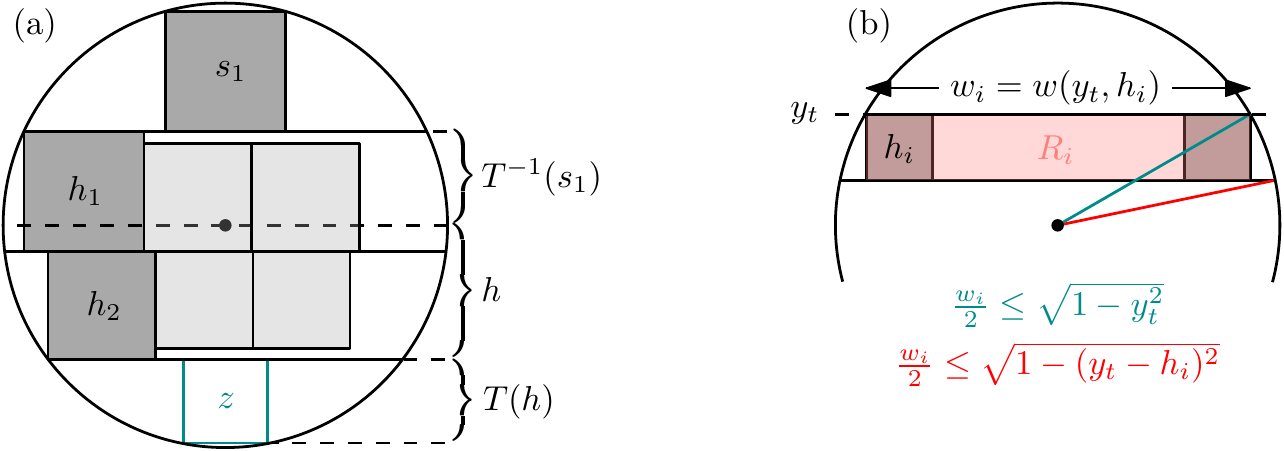}
	\caption{(a) Computing $z$ for $k = 2$ using functions $T$ and $T^{-1}$. (b) Computing the width $w_i = w(y_t,h)$ of rectangle $R_i$.}
	\label{fig:formula-z}
\end{figure*}
We consider a sequence $s_1,\ldots,s_n$ of squares of total area $S$ that \algo fails to pack, and assume w.l.o.g.\ that  $s_{n-1}$ is packed.
This implies $z < s_n$, where $z$ denotes the side length of the largest square that could be packed below the last subcontainer $C_k$ constructed by \BP;
otherwise, a further subcontainer is constructed.
We have $z = T(-T^{-1}(s_1) + \sum_{i=1}^{k} h_i)$; see \cref{fig:formula-z}(a).

By $w(y_t,h) = 2\sqrt{\min\{1 - y_t^2, 1 - (y_t-h)^2\}}$,
we denote the maximum width of a rectangle~$R$ that can be placed in $\mathcal{D}$ with top side at $y = y_t$ and height $h$; see \cref{fig:formula-z}(b).
Thus, we can express the width $w_i$ of the rectangle $R_i$ inscribed in some subcontainer $C_i$ in terms of $s_1,h_1,\ldots,h_i$ as
\[ w_i \coloneqq w\left(\bigg( T^{-1}(s_1)-\sum\limits_{j=1}^{i-1}h_j\bigg), h_i\right).\]

Recall that $\sigma \coloneqq \sigma(s_1)$ denotes the side length of the largest squares that fits into the pockets $\CL$ and $\CR$ as illustrated in \cref{fig:t-t1-sigmaApp}(c).
In order to distinguish whether \TP has packed any squares into the pockets, we consider the function
\[
E(s_1, s_n) \coloneqq 
\begin{cases}
0.83\cdot \sigma(s_1)^2,& \text{if } s_n \leq \sigma,\\
0,& \text{ otherwise,}
\end{cases}
\]
which describes the total square area that \TP is guaranteed to pack due to \cref{lem:anapocketpacking}.

For the analysis of Case~\CIII, we distinguish cases depending on the number $k$ of subcontainers constructed by \BP.
Specifically, we consider the cases $k=0,k=1,k\in\{2,3,4\}$, and $k \geq 5$.

\subsubsection{Analysis for no subcontainer}
\label{sec:NoRow}
\begin{restatable}{lemma}{anaNoStrip}\label{lem:anaNoStrip}
	If \algo fails to pack a sequence $s_1,\ldots,s_n$ of squares and \BP does not construct a subcontainer, the total area of the squares exceeds $\nicefrac{8}{5}$.
\end{restatable}
\begin{proof}
	Because the algorithm fails to construct a first subcontainer in the bottom part, it follows that placing $s_n$ as far to the bottom as possible yields an overlap with $s_1$. 
	However, the minimum value for $s_1^2 + s_n^2$ for two overlapping squares packed into a disk is attained for $s_1 = s_n$.
	This corresponds to the worst-case configuration, implying that the total area of $s_1$ and $s_n$ exceeds $\nicefrac{8}{5}$.
\end{proof}

\subsubsection{Analysis for one subcontainer}
\begin{restatable}{lemma}{anaOneRow}\label{lem:anaOneRow}
	If \algo fails to pack a sequence $s_1,\ldots,s_n$ of squares and \BP constructs exactly one subcontainer, the total area of the squares exceeds $\nicefrac{8}{5}$.
\end{restatable}
\begin{proof}
	Combining \cref{lem:anapocketpacking,cor:combined} allows us to bound the area of $s_1,\ldots,s_n$ by
	\[ S \geq F_{SC_1}(s_1,h_1,s_n) \coloneqq s_1^2 + B_{4}(T^{-1}(s_1), h_1, w_1, s_n) + s_n^2 + E(s_1, s_n),\]
	where $
	E(s_1, s_n) =
	0.83\cdot \sigma(s_1)^2 $ if  $s_n \leq \sigma$, and 
	$E(s_1, s_n)=0$ if  $s_n > \sigma$.
	Furthermore, we know $0 < z < s_n$, because \algo fails to pack $s_n$.
	Moreover, we claim that at least one of the following conditions must hold:
	$s_1 > \nicefrac{1}{\sqrt{2}}$, $w_1 < 2h_1$, or $s_1^2 + h_1^2 + 2s_n^2 < \nicefrac{39}{25}$.
	Assume for contradiction that neither of these conditions hold.
	By $w_1 \geq 2h_1$, we know that at least two squares are packed into the first subcontainer.
	One of these squares has area $h_1^2$, and the other has area at least $s_n^2$.
	In particular, this implies that the algorithm packs $s_1,s_2$ and $s_3$.
	This implies $s_2 \geq h_1$ and $s_3,s_4 \geq s_n$.
	Thus we have $s_1^2 + s_2^2 + s_3^2 + s_4^2 \geq s_1^2 + h_1^2 + 2s_n^2 \geq \nicefrac{39}{25}$,
	which together with $s_1 \leq \nicefrac{1}{\sqrt{2}}$ implies that we are in Case~\CII of our algorithm.
	This is a contradiction, because we only construct subcontainers in Case~\CIII.
	Thus the following lemma, proved automatically using interval arithmetic, proves that these conditions are sufficient to ensure $S \geq \nicefrac{8}{5}$.
	\begin{lemma}[\textsc{One Subcontainer}, Automatic Analysis for Lemma~\ref{lem:anaOneRow}]\label{lem:summaryInterval2}
		Let $z \coloneqq T(T^{-1}(s_1) + h_1)$.
		For all $s_1, h_1, s_n$ with $0 < z < s_n \leq h_1 \leq s_1$, $h_1 \leq T^{-1}(s_1) + 1$ and
		\[\left(s_1 > \nicefrac{1}{\sqrt{2}}\right) \vee \left(w_1 < 2h_1\right) \vee \left(s_1^2 + h_1^2 + 2s_n^2 < \nicefrac{39}{25}\right),\]
		we have $F_{SC_1}(s_1,h_1,s_n) > \nicefrac{8}{5}$.\qedhere
	\end{lemma}
\end{proof}

\subsubsection{Analysis for two to four subcontainers}
\begin{restatable}{lemma}{anaTwoThreeFourRows}\label{lem:anaTwoThreeFourRows}
	If \algo fails to pack a sequence $s_1,\ldots,s_n$ of squares and \BP constructs $k \in \{2,3,4\}$ subcontainers, the total area of the squares exceeds $\nicefrac{8}{5}$.
\end{restatable}
\begin{proof}
	We use similar ideas as in the proof of \cref{lem:anaOneRow}.
	We bound the area packed by \TP by $E(s_1, s_n)$ using \cref{lem:anapocketpacking}.
	Furthermore, we use \cref{cor:combined} to bound the area packed into each of the $k \in \{2,3,4\}$ subcontainers by
	\[ \| C_i \| \geq B_4\left(\bigg(T^{-1}(s_1) - \sum\limits_{j=1}^{i-1} h_i\bigg), h_i, w_i, h_{i+1}\right), 1 \leq i \leq k,\]
	where $h_{k+1} \coloneqq s_n$.
	We can express $w_i = w(T^{-1}(s_1) - \sum_{j=1}^{i-1}h_j, h_i)$ in terms of $s_1$ and $h_j, 1 \leq j \leq i$.
	In total, for $k$ subcontainers, this yields the bound
	\[ S \geq F_{SC_k}(s_1,h_1,\ldots,h_{k},s_n) \coloneqq s_1^2 + s_n^2 + E(s_1, s_n) + \sum\limits_{i=1}^{k} B_4\left(\bigg(T^{-1}(s_1) - \sum\limits_{j=1}^{i-1} h_i\bigg), h_i, w_i, h_{i+1}\right).\]
	Finally, we know $0 < z < s_n$ because the algorithm fails to pack $s_n$.
	Thus, the following lemma, proved automatically using interval arithmetic, suffices to complete the proof of \cref{lem:anaTwoThreeFourRows}.
	\begin{lemma}[Automatic Analysis for Lemma~\ref{lem:anaTwoThreeFourRows}]
		Let $z_k = T(-T^{-1}(s_1) + \sum_{i=1}^{k}h_i)$.
		\begin{description}
			\item[$(k = 2)$] For all $s_1,h_1,h_2,s_n$ with $0 < z_2 < s_n \leq h_2 \leq h_1 \leq s_1$ and $0.295 \leq s_1 \leq \sqrt{\nicefrac{8}{5}}$ and $h_1 + h_2 \leq 1 + T^{-1}(s_1)$,
			we have $F_{SC_2} > \nicefrac{8}{5}$.
			\item[$(k = 3)$] For all $s_1,h_1,h_2,h_3,s_n$ with $0 < z_3 < s_n \leq h_3 \leq h_2 \leq h_1 \leq s_1$
			and $0.295 \leq s_1 \leq \sqrt{\nicefrac{8}{5}}$ and $h_1 + h_2 + h_3 \leq 1 + T^{-1}(s_1)$,
			we have $F_{SC_3} > \nicefrac{8}{5}$.
			\item[$(k = 4)$] For all $s_1,h_1,h_2,h_3,h_4,s_n$ with $0 < z_4 < s_n \leq h_4 \leq h_3 \leq h_2 \leq h_1 \leq s_1$
			and $0.295 \leq s_1 \leq \sqrt{\nicefrac{8}{5}}$ and $h_1 + h_2 + h_3 + h_4 \leq 1 + T^{-1}(s_1)$,
			we have $F_{SC_4} > \nicefrac{8}{5}$.\qedhere
		\end{description}
	\end{lemma}
\end{proof}

We defer the analysis for five or more subcontainers to \cref{sec:details-analysis-algorithm}.
This completes the analysis of \CIII and thus the proof of our main result.

\subsubsection{Analysis for five or more subcontainers}
\label{sec:details-analysis-algorithm}
In this section, we handle the case where \algo constructs at least five subcontainers.
\begin{restatable}{lemma}{anaFiveRows}\label{lem:anaFiveRows}
	If \algo fails to pack a sequence $s_1,\ldots,s_n$ of squares and \BP constructs $k \geq 5$ subcontainers, the total area of the squares exceeds $\nicefrac{8}{5}$.
\end{restatable}
On the highest level, we distinguish the two cases 
$s_n > \sigma$ (\cref{lem:five-containers-above-sigma})
and  $s_n \leq \sigma$ (\cref{lem:five-containers-below-sigma}), together establishing \cref{lem:anaFiveRows}.

\paragraph{Proof of \texorpdfstring{\cref{lem:five-containers-above-sigma}}{Claim 22}}
\begin{claim}
	\label{lem:five-containers-above-sigma}
	If \algo fails to pack a sequence $s_1,\ldots,s_n$ of squares, $s_n > \sigma$ and
	\BP constructs $k \geq 5$ subcontainers, the total area of the squares exceeds $\nicefrac{8}{5}$.
\end{claim}
\begin{claimproof}
	Because of $0.295 \leq s_1 \leq \sqrt{\nicefrac{8}{5}}$ \new{and \cref{eq:sigma}}, we have $\sigma > 0.231$.
	Of the $2$ units of vertical space, $s_1$ takes up $1 - T^{-1}(s_1) > 0.305$.
	Therefore, $s_n > \sigma$ implies that only $k \leq \lfloor \nicefrac{2-0.305}{0.231}\rfloor = 7$ subcontainers can be constructed. 
	Because $s_n > \sigma$, no square is packed into the pockets by \TP.
	We use \cref{cor:combined} to derive the same lower bound on the area packed into container $C_i$ as in the proof of \cref{lem:anaTwoThreeFourRows}:
	\[ \| C_i \| \geq B_4\left(\bigg(T^{-1}(s_1) - \sum\limits_{j=1}^{i-1} h_i\bigg), h_i, w_i, h_{i+1}\right), 1 \leq i \leq k,\]
	where $h_{k+1} \coloneqq s_n$.
	We then bound the total area of $s_1,\ldots,s_n$ by 
	\[S \geq F_{SC_k}(s_1,h_1,\ldots,h_{k},s_n) \coloneqq s_1^2 + s_n^2 + \sum\limits_{i=1}^{k} B_4\left(\bigg(T^{-1}(s_1) - \sum\limits_{j=1}^{i-1} h_i\bigg), h_i, w_i, h_{i+1}\right).\]
	Thus, the following lemma proved using interval arithmetic concludes the proof, and thus the proof of \cref{lem:anaFiveRows} for the case $s_n > \sigma$.
\end{claimproof}

\begin{lemma}[Automatic Analysis for \cref{lem:five-containers-above-sigma}]
	Let $z_k = T(-T^{-1}(s_1) + \sum_{i=1}^{k}h_i)$.
	\begin{description}
		\item[$(k = 5)$] For all $s_1,h_1,\ldots,h_5,s_n$ with $0 < z_5,\sigma < s_n \leq h_5 \leq \cdots \leq h_1 \leq s_1$
		and $0.295 \leq s_1 \leq \sqrt{\nicefrac{8}{5}}$ and $h_1 + \cdots + h_5 \leq 1 + T^{-1}(s_1)$,
		we have $F_{SC_5} > \nicefrac{8}{5}$.
		\item[$(k = 6)$] For all $s_1,h_1,\ldots,h_6,s_n$ with $0 < z_6,\sigma < s_n \leq h_6 \leq \cdots \leq h_1 \leq s_1$
		and $0.295 \leq s_1 \leq \sqrt{\nicefrac{8}{5}}$ and $h_1 + \cdots + h_6 \leq 1 + T^{-1}(s_1)$,
		we have $F_{SC_6} > \nicefrac{8}{5}$.
		\item[$(k = 7)$] For all $s_1,h_1,\ldots,h_7,s_n$ with $0 < z_7,\sigma < s_n \leq h_7 \leq \cdots \leq h_1 \leq s_1$
		and $0.295 \leq s_1 \leq \sqrt{\nicefrac{8}{5}}$ and $h_1 + \cdots + h_7 \leq 1 + T^{-1}(s_1)$,
		we have $F_{SC_7} > \nicefrac{8}{5}$.
	\end{description}
\end{lemma}

\paragraph{Proof of \texorpdfstring{\cref{lem:five-containers-below-sigma}}{Claim 24}}
\begin{claim}
	\label{lem:five-containers-below-sigma}
	If \algo fails to pack a sequence $s_1,\ldots,s_n$ of squares, $s_n \leq \sigma$ and 
	\BP constructs $k \geq 5$ subcontainers, the total area of the squares exceeds $\nicefrac{8}{5}$.
\end{claim}

In the last remaining case, we have $k \geq 5$ subcontainers and $s_n \leq \sigma$.
To handle the \new{arbitrarily large number} of subcontainers in this case,
we begin by deriving bounds for the total area of squares packed into subcontainers $C_j,\ldots,C_k$
and $s_n$, which we assume, w.l.o.g., to be the first square that \algo fails to pack.

Let $C_j$ be a subcontainer whose top side is below the center of $\mathcal{D}$, i.e., below the line $y = 0$.
Let $A_{j}$ be the area of the part of the disk below the top side of $C_{j}$,
containing $C_{j},\dots,C_k$; see \cref{fig:rowsUpperBoundApp}(a).
Let $H_j$ be the vertical distance between the top side of $C_j$ and the lowest point of $\mathcal{D}$.

\begin{figure*}[htb]
	\centering
	\includegraphics[page=3]{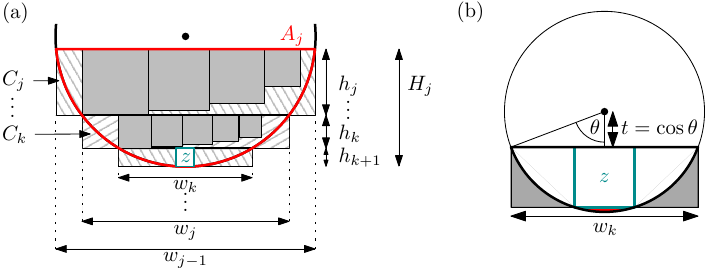}
	\caption{
		(a)~The definition of $A_j, h_j, H_j$ and $w_j$.
		The area of a subcontainer $C_i$, $i\geq j$ can be upper bounded by the area of the smallest enclosing rectangle (gray-white hatched).
		The area below $C_k$ is upper bounded by the area of the $z \times w_k$-rectangle (gray-white hatched) below $C_k$,
		where $z$ denotes the side length of a largest square fitting below $C_k$.
		(b)~Bounding the area below $C_k$  (with bold outline) by $zw_k < h_{k+1}w_k$ requires
		showing that the red area does not exceed the gray area.}
	\label{fig:rowsUpperBoundApp}
\end{figure*}

We first prove the following lemma.
\begin{lemma}\label{lem:totalAreaPackedByBasePacking}
	If $s_n \leq 0.6$, the total area of $s_n$ and the squares packed into $C_j,\ldots,C_k$ by \BP is at least
	\[B_5(h_j,H_j) \coloneqq A_{j+1} + h_j^2 - H_jh_j.\]
\end{lemma}
\begin{proof}
	Applying \cref{lem:anaFramesOfBasePacking} to $C_j,\ldots,C_k$ yields that
	the total area of $s_n = h_{k+1}$ and the squares packed into $C_j,\ldots,C_k$ is at least
	\begin{align*}
	S :=h_{k+1}^2+\sum\limits_{i=j}^{k} \| C_i \| &\geq h_{k+1}^2 + \sum\limits_{i=j}^{k} \big( h_i^2 + h_{i+1}\cdot(w_i - h_i - h_{i+1})\big)\\
	&= h_{k+1}^2 + \sum\limits_{i=j}^{k} h_i^2 - \sum\limits_{i=j}^{k}h_{i+1}^2 + \sum\limits_{i=j}^{k} h_{i+1}w_i - \sum\limits_{i=j}^{k} h_{i+1}h_i \\
	&= h_{k+1}^2 + h_j^2 - h_{k+1}^2 + \sum\limits_{i=j}^{k} h_{i+1}w_i - \sum\limits_{i=j}^{k} h_{i+1}h_i\\
	&= h_j^2 + \sum\limits_{i=j}^{k} h_{i+1}w_i - \sum\limits_{i=j}^{k} h_{i+1}h_i.
	\end{align*}
	
	Next, we show that we can bound the disk area $A_{k+1}$ below $C_k$ by $B \coloneqq w_kh_{k+1}$;
	see \cref{fig:rowsUpperBoundApp}(b).
	We have $A_{k+1} = \theta-t\sqrt{1-t^2} = \theta - \sin\theta \cdot \cos\theta$, where $t \coloneqq \cos\theta$.
	We also have $(t+z)^2 + \nicefrac{z^2}{4} = 1$, i.e., $t = \sqrt{1-\nicefrac{z^2}{4}} - z$.
	Because $0 < z < s_n \leq 0.6 < 1$, we can bound this by 
	\[ t > 1 - \nicefrac{1}{4}\cdot z^2- z > 1 - \nicefrac{5}{4}\cdot z \implies z > \nicefrac{4}{5}\cdot(1-t).\]
	We have $B = h_{k+1}w_k > zw_k = z \cdot 2\sin\theta$, and thus 
	$B > \nicefrac{4}{5}\cdot(1-\cos\theta) \cdot 2\sin\theta = \nicefrac{8}{5} \cdot (1-\cos\theta)\sin\theta,$
	which implies \[B - A_{k+1} > \nicefrac{8}{5}\cdot\sin\theta - \nicefrac{8}{5}\cdot\sin\theta\cos\theta-\theta+\sin\theta\cos\theta = \underbrace{\nicefrac{8}{5}\sin\theta - \nicefrac{3}{5}\sin\theta\cos\theta - \theta}_{\eqqcolon f(\theta)}.\]
	Clearly, $f(0) = 0$.
	Moreover, the derivative is 
	\[ f'(\theta)
	= -\nicefrac{2}{5}(3\cos^2\theta-4\cos\theta+1) 
	= \nicefrac{2}{5}(1-\cos\theta)(3\cos\theta - 1),\] 
	which is non-negative for $0 \leq \theta \leq \cos^{-1}(\nicefrac{1}{3})$.
	Hence, we have $f(\theta) \geq 0$ and thus $B > A_{k+1}$ for $0 \leq \theta \leq \cos^{-1}(\nicefrac{1}{3})$.
	Because of $z < s_n \leq 0.6$, the largest angle $\theta$ that occurs is
	\[\theta = \underbrace{\cos^{-1}\left(\sqrt{1-\nicefrac{0.6^2}{4}}-0.6\right)}_{\approx 1.209} <
	\underbrace{\cos^{-1}\left(\nicefrac{1}{3}\right)}_{\approx 1.230}.\]
	
	We have thus established $A_{k+1} < h_{k+1}w_k$ and can bound 
	$A_{j+1} \leq \sum_{i=j+1}^{k+1}h_iw_{i-1} = \sum_{i=j}^{k}h_{i+1}w_i$;
	see \cref{fig:rowsUpperBoundApp}(a).
	Thus, by $h_j \geq h_{j+1} \geq \cdots \geq h_{k+1} = s_n$, we have
	\begin{align*}
	S &\geq h_j^2 + A_{j+1} - \sum\limits_{i=j}^{k}h_{i+1}h_i \geq h_j^2 + A_{j+1} - h_j\sum\limits_{i=j}^{k}h_{i+1}\\
	&\geq h_j^2 + A_{j+1} - h_j\sum\limits_{i=j}^{k}h_{i} \geq h_j^2 + A_{j+1} - h_jH_j\text{,}
	\end{align*}
	and thus the claimed bound.
\end{proof}

To apply \cref{lem:totalAreaPackedByBasePacking}, we first observe that $\sigma(s_1) \leq 0.6$ for any $s_1$.
Assume for contradiction that $\sigma(s_1) > 0.6$ for some $s_1$.
Because two squares of side length $\sigma$ fit into $\mathcal{D}$ besides $s_1$,
$2\sigma + s_1 < 2$ and thus $s_1 < 0.8$.
However, for $0 \leq s_1 \leq 0.8$, $\sigma(s_1)$ is monotonically increasing with $\sigma(0.8) < 0.4345 < 0.6$, which is a contradiction.

To prove \cref{lem:five-containers-below-sigma}, 
we distinguish whether the $y$-coordinate $y_3$ of the bottom side of
the third subcontainer $C_3$ is positive or not.
If $y_3 \leq 0$, i.e., $C_3$'s bottom side is below the center of $\mathcal{D}$,
we apply \cref{lem:totalAreaPackedByBasePacking,lem:anapocketpacking,lem:Aeleven} to lower-bound the total square area of $s_1,\ldots,s_n$ by 
\[
F_{MSC_1} \coloneqq s_1^2 + 0.83\sigma^2 + B_5(h_4, H_4) + \sum_{i=1}^{3} B_4\left(\bigg(T^{-1}(s_1)-\sum_{\ell=1}^{i-1}h_{\ell}\bigg), h_i, w_i, h_{i+1}\right),
\]
with $H_4 = 1 + T^{-1}(s_1) - h_1 - h_2 - h_3$ and $w_i = w(T^{-1}(s_1) - \sum_{\ell=1}^{i-1}h_{\ell}, h_i)$.
The case $y_3 \leq 0$ can then be handled using the following lemma, proved by interval arithmetic.
\begin{lemma}[Automatic Analysis for \cref{lem:five-containers-below-sigma}, $y_3 \leq 0$]
	For all $s_1 \geq h_1 \geq \cdots \geq h_4 > 0$
	with $0.295 \leq s_1 \leq \sqrt{\nicefrac{8}{5}}$ and
	$0 \leq H_4 \leq 1$, we have $F_{MSC_1} > \nicefrac{8}{5}$.
\end{lemma}

In the last remaining case we have $y_3 > 0$, i.e., the third subcontainer lies strictly above the center of $\mathcal{D}$.
Let $j$ be the maximal index such that the top side of $C_j$ lies strictly above $\mathcal{D}$'s center; we have $j > 3$ due to $y_3 > 0$.
As before, we bound the area packed by \TP by $s_1^2 + 0.83\sigma^2$ with \cref{lem:anapocketpacking} and the area packed by \BP into $C_1$ and $C_2$ using \cref{cor:combined}.
Furthermore, we bound the area packed into $C_{j+1},\ldots,C_k$ and $s_n$ using \cref{lem:totalAreaPackedByBasePacking}.

For the total square area $S_c$ packed into $C_3,\ldots,C_j$, 
we consider the rectangle $R$ of height $H_R = \sum_{i=3}^{j}h_i$ and 
width $W_R = \min\{w_3,\ldots,w_j\}$ inscribed into $C_3 \cup \cdots \cup C_j$;
see \cref{fig:y3-gt-0}.
Next, we prove the following bound on the total square area $S_c$ packed into $C_3,\ldots,C_j$:
\[ S_c \geq B_6(H_R,W_R,h_{j+1}) \coloneqq \nicefrac{1}{2}\cdot H_RW_R + \nicefrac{1}{4}\cdot h_{j+1}H_R.\]

\begin{figure}[htb]
	\centering
	\includegraphics[scale=1]{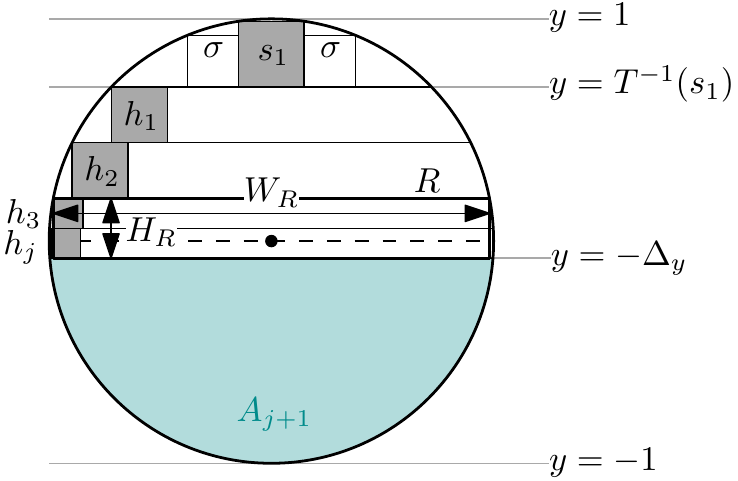}
	\caption{In the case $y_3 > 0$, we use the rectangle $R$ of height $H_R$ and width $W_R$
		inscribed in $C_3,\ldots,C_j$ to bound the total area packed into $C_3,\ldots,C_j$.}
	\label{fig:y3-gt-0}
\end{figure}

First, we observe that due to the fact that three subcontainers fit above $\mathcal{D}$'s center and below~$s_1$, $T^{-1}(s_1) - h_1 - h_2 - h_3 > 0$ and thus $h_3 < \nicefrac{1}{3}\cdot T^{-1}(s_1) \leq \nicefrac{1}{3} \cdot T^{-1}(0.295)< 0.232$.
This also bounds the distance $\Delta_y$ between $\mathcal{D}$'s center and the bottom side of $R$ by $0.232$.
We thus obtain $W_R \geq \min(2\sqrt{1-T^{-1}(0.295)^2}, 2\sqrt{1-0.232^2}) > 1.439 > 2h_3$.
This yields $w_i > 2h_i$ for all $i \in \{3,\ldots,j\}$.
This allows us to apply \cref{lem:shelfpackingnew},
which yields \[\|C_i\| \geq \nicefrac{1}{2}\cdot h_iw_i + \nicefrac{1}{4}\cdot h_i^2 \geq \nicefrac{1}{2}\cdot h_iw_i + \nicefrac{1}{4}\cdot h_ih_{j+1}\] for $i \in \{3,\ldots,j\}$.
Summing over $i$ yields
\[ S_c \geq \sum\limits_{i=3}^{j} \frac{1}{2}h_iw_i + \sum\limits_{i=3}^{j} \frac{1}{4}h_ih_{j+1}
\geq \left(\frac{1}{2}W_R + \frac{1}{4}h_{j+1}\right)\sum\limits_{i=3}^{j}h_i = B_6(H_R,W_R,h_{j+1}),\] as claimed.
Overall, we can thus bound the total area $S$ of $s_1,\ldots,s_n$ by 
\begin{align*}
S \geq F_{MSC_2} &\coloneqq s_1^2 + 0.83\sigma^2 + B_4(T^{-1}(s_1), h_1, w_1, h_2)\\
& + B_4(T^{-1}(s_1)-h_1,h_2,w_2,h_3) + B_6(H_R, W_R, h_{j+1}) + B_5(h_{j+1}, H_{j+1}).
\end{align*}
We have $H_R = T^{-1}(s_1) - h_1 - h_2 + \Delta_y$ and $W_R = w(T^{-1}(s_1) - h_1 - h_2, H_R)$ and $H_{j+1} = 1 - \Delta_y$.
Therefore, the following lemma, proved using interval arithmetic, concludes the proof of \cref{lem:five-containers-below-sigma} and thus the proof of \cref{lem:anaFiveRows}.
\begin{lemma}[Automatic Analysis for \cref{lem:five-containers-below-sigma}, $y_3 > 0$]
	For all $s_1 \geq h_1 \geq h_2 \geq h_3 \geq h_{j+1}$,
	$0.295 \leq s_1 < \sqrt{\nicefrac{8}{5}}$ and all $\Delta_y$ with
	$T^{-1}(s_1) - h_1 - h_2 - h_3 > 0$ and $\Delta_y \leq h_3$, we have $F_{MSC_2} > \nicefrac{8}{5}$.
\end{lemma}

\section{Conclusion \new{and future directions}}
\label{sec:conc}
We have established the critical density for packing squares into a disk: 
Any set of squares of total area at most $\nicefrac{8}{5}$ can be packed into a unit disk. As shown by our lower bound example, this guarantee is best-possible, i.e., it cannot be improved.
The proof is based on an algorithm that subdivides the disk into horizontal subcontainers and uses a refined shelf packing scheme. 
The correctness of this algorithm is shown by careful manual analysis, complemented by a computer-assisted part that is based on interval arithmetic.

There is a variety of interesting directions for future research.
Of particular interest is the critical density for packing squares of bounded size into a disk,
which will result in a higher packing density;
a more general problem concerns the critical packing density 
for packing other types of objects of bounded size into other types of containers.
Other questions arise from considering questions in three- or even higher-dimensional space.
We are optimistic that many of our techniques will be useful for settling these problems.

\bibliography{references}
\end{document}